%% file: TSP_v01.tex
\documentclass[conference]{IEEEtran}

\newif\ifFullVersion
\usepackage{cite}
\usepackage{amsmath,amssymb,amsfonts}
\usepackage{algorithmic}
\usepackage{graphicx}
\usepackage{textcomp}
\usepackage{xcolor}
\def\BibTeX{{\rm B\kern-.05em{\sc i\kern-.025em b}\kern-.08em
    T\kern-.1667em\lower.7ex\hbox{E}\kern-.125emX}}

\usepackage{acronym}
\usepackage[belowskip=-10pt,aboveskip=1.8pt]{caption}
\usepackage{standalone}
\usepackage{subcaption}
\usepackage{tikz}
\usepackage{url,enumitem, cite}
\usepackage{verbatim}
\usepackage[bookmarks,colorlinks]{hyperref}
\usepackage{soul, xcolor}
\usepackage{mathrsfs}
\usepackage[ruled,linesnumbered,vlined]{algorithm2e}
\usepackage{dsfont}
\SetKwInput{KwData}{\textbf{Init}} 

\usepackage{tikz}
\usepackage{pgfplots}
\usepackage{xcolor}
\pgfplotsset{compat=1.10}
\usepgfplotslibrary{groupplots}
\usetikzlibrary{calc}
\usepackage{float}
\usepackage{placeins}

\usepackage[all=normal,paragraphs=tight,floats=normal,mathspacing=normal,wordspacing=tight,charwidths=tight,mathdisplays=normal,leading=normal]{savetrees}
\usepackage{pgfplots}
\usepackage{pgfplotstable}
\usetikzlibrary{spy}
\pgfplotsset{compat=newest}
\usepackage{booktabs}
\usepackage{csvsimple}

\usepackage{amsmath,amssymb,mathtools}
\usepackage{empheq} 

\usepackage{booktabs,tabularx}
\newcommand{\myVec}[1]{{\boldsymbol{#1}}}
\newcommand{\myMat}[1]{{\boldsymbol{#1}}}
\newcommand{\mySet}[1]{\mathcal{#1}}

\newcommand{\figref}[1]{Fig.~\ref{#1}}

\newcommand{\sbrackets}[1]{\left[#1\right]}
\newcommand{\expecteds}[1]{\mathds{E}\sbrackets{#1}}
\let\oldnl\nl
\newcommand{\nonl}{\renewcommand{\nl}{\let\nl\oldnl}}

\acrodef{lbg}[LBG]{Linde–Buzo–Gray}

\newtheorem{theorem}{Theorem}

\newtheorem{corollary}{Corollary}

\newcommand{\SigH}{\boldsymbol\Sigma^{\rm (H)}}
\newcommand{\SigT}{\boldsymbol\Sigma^{\rm (T)}}

\newcommand{\TSigH}{\tilde{\boldsymbol\Sigma}^{\rm (H)}}
\newcommand{\TSigT}{\tilde{\boldsymbol\Sigma}^{\rm (T)}}

\acrodef{ai}[AI]{artificial intelligence}
\acrodef{dl}[DL]{deep learning}
\acrodef{bs}[BS]{base station}
\acrodef{dnn}[DNN]{deep neural network}
\acrodef{cnn}[CNN]{convolutional neural network}
\acrodef{mlp}[MLP]{multi-layer perceptron}
\acrodef{snr}[SNR]{signal-to-noise ratio}
\acrodef{awgn}[AWGN]{additive white Gaussian noise} 
\acrodef{ml}[ML]{machine learning} 
\acrodef{sgd}[SGD]{stochastic gradient descent} 
\acrodef{mse}[MSE]{mean-squared error}
\acrodef{rmse}[RMSE]{root mean squared error}
\acrodef{rmspe}[RMSPE]{root mean squared periodic error}
\acrodef{mle}[MLE]{maximum likelihood estimation}
\acrodef{snr}[SNR]{signal-to-noise ratio}
\acrodef{admm}[ADMM]{alternating direction method of multipliers}
\acrodef{aoa}[AoA]{Angle of Arrival}
\acrodefplural{aoa}[AoAs]{Angles of Arrival}
\acrodef{em}[EM]{electromagnetic}
\acrodef{cmos}[CMOS]{complementary metal-oxide semiconductor}
\acrodef{ula}[ULA]{uniform linear array}
\acrodef{em}[EM]{Electromagnetic}
\acrodef{doa}[DoA]{direction of arrival}
\acrodef{music}[MUSIC]{MUltiple SIgnal Classification}
\acrodef{esprit}[ESPRIT]{Estimation of signal parameters via rotational invariance techniques}
\acrodef{evd}[EVD]{eigenvalues decomposition}
\acrodef{sps}[SPS]{spatial smoothing}
\acrodef{iid} [i.i.d] {Independent and identically distributed}
\acrodef{ls}[LS]{Least Square}
\acrodef{relu}[ReLu]{Rectified Linear Unit}
\acrodef{crb} [CRB] {Cramér–Rao Bound}
\acrodef{ccrb} [CCRB] {Conditional Cramér–Rao Bound} 
\acrodef{vqvae}[VQ-VAE]{vector quantized variational autoencoder}
\acrodef{rm}[Root-MUSIC]{Root-MUSIC}
\acrodef{esprit}[ESPRIT]{Estimation of Signal Parameters via Rotational Invariance Techniques}
\acrodef{drm}[DR-MUSIC]{Deep Root-MUSIC}
\acrodef{ssn}[SubspaceNet]{Subspace Net}
\acrodef{rssn}[Remote SubspaceNet]{Remote Subspace Net}
\acrodef{mbdl}[Model Based Deep Learning]{Model Based Deep Learning}
\acrodef{anees}[ANEES]{average normalized estimation error squared}

\setstcolor{blue}

\input{plotting_latex_code/benchmark_plot_common.tex}

\input{plotting_latex_code/angle_sweep_setup_snr0.tex}

\IEEEoverridecommandlockouts
\begin{document}

\title{AI-Aided ESPRIT for Joint DoA\\ Estimation and Uncertainty Extraction
}

\author{
	\IEEEauthorblockN{Raz Zohar and Nir Shlezinger
\thanks{ 
The preliminary findings of this research were presented in the IEEE Sensor Array and Multichannel Signal Processing Workshop (SAM 2026) as the paper \cite{zohar2026deep}.  
This work was support by the Israel Science Foundation (ISF) under grant no. 3314/25, and by the Israeli Ministry of Science and Technology. 
R. Zohar and N. Shlezinger are with the ECE School, Ben-Gurion University of the Negev, Israel (e-mails: razzoh@post.bgu.ac.il;
nirshl@bgu.ac.il).  
}}
}

\maketitle


\begin{abstract} 
\Ac{doa} estimation often requires not only accurate recovery of source directions, but also reliable characterization of the uncertainty in these estimates. While classical subspace methods such as \ac{esprit} provide principled uncertainty analyses, their performance and uncertainty quantification rely on restrictive assumptions. Recent deep learning approaches alleviate these limitations and enable robust \ac{doa} estimation in challenging  conditions, but generally provide point estimates and lack principled uncertainty characterization. In this work, we develop an \ac{ai}-aided framework for joint \ac{doa} estimation and uncertainty quantification that combines the robustness of model-based deep learning with the analytical foundations of classical subspace methods. Building on \ac{ai} surrogate covariance recovery, we extend existing \ac{esprit} uncertainty analyses to characterize the full covariance structure of the \ac{doa} estimation error and integrate this characterization into a subspace-oriented deep learning architecture. We further propose a dedicated learning strategy that jointly promotes accurate \ac{doa} recovery and faithful uncertainty estimation. The resulting methodology preserves the interpretable processing pipeline of classical subspace methods while enabling reliable operation in regimes where conventional approaches struggle. Our numerical studies demonstrate that the proposed framework consistently achieves accurate \ac{doa} estimation together with reliable uncertainty characterization across diverse challenging scenarios, including coherent sources, limited snapshots, and array calibration errors.
\end{abstract}

\acresetall

\section{Introduction}
\label{sec:intro} 
\Ac{doa} estimation is a key task in array signal processing, playing a central role in radar, sonar, wireless communications, acoustics, radio astronomy, and localization systems~\cite{pillai2012array}. By exploiting the spatial diversity induced by sensor arrays, \ac{doa} estimation algorithms infer the angular location of signal emitters from noisy observations collected across multiple sensors~\cite{tuncer2009classical}. In various array processing applications, accurate angle recovery alone is insufficient, requiring not only point estimates of the source directions, but also a characterization of the confidence associated with these estimates~\cite{haykin1992some}. Such uncertainty information is critical for downstream tasks including target tracking~\cite{konstantino2026unsupervised}, sensor fusion~\cite{hawkes2003wideband}, channel estimation~\cite{yang2020bayesian}, and signal enhancement~\cite{lam2006bayesian},  where the reliability of the estimated parameters can be as important as the estimates themselves.

 Classical beamforming approaches~\cite{capon1969mvdrbf} estimate source directions by scanning the angular domain and identifying peaks in a spatial spectrum, with an angular resolution that is limited by the array aperture and observation conditions~\cite{benesty2017fundamentals}. To overcome these limitations, subspace-based methods exploit the \ac{evd} of the array covariance matrix to separate the signal and noise subspaces. Representative examples include \ac{music}~\cite{schmidt1986music} and Root \ac{music}~\cite{Barabell1983ImprovingTR}, which identifies directions through subspace orthogonality, as well as  \ac{esprit}~\cite{roy1989esprit}, which leverages the rotational invariance structure of the array manifold to  recover the source angles. Subspace methods constitute a leading family of \ac{doa} estimation methods due to their high resolution, strong theoretical foundations, and interpretable processing pipeline~\cite{liu2023twenty}.

 Beyond angle recovery,  array processing theory provides principled tools for uncertainty characterization. These include the characterization of fundamental performance limits, such as the \ac{crb}~\cite{stoica1989music,liang2020review}, that are invariant of the specific estimator employed, as well as algorithm-specific uncertainty quantification based on  perturbation analyses of subspace-based estimators such as Root-\ac{music}~\cite{rao2002performance} and \ac{esprit}~\cite{yuen2002asymptotic}. The latter methods enable the prediction of estimation variance  for a given estimate  from the observed data. However, subspace methods and their corresponding uncertainty analyses are derived under restrictive assumptions that are often violated in practice, including sufficiently many independent snapshots, accurate array calibration,  and non-coherent source emissions. Furthermore, uncertainty extraction is often more sensitive to deviations than \ac{doa} estimation itself, as reliable variance prediction requires not only accurate angle recovery but also a faithful characterization of the estimator statistics. Consequently, in challenging operating regimes involving coherent sources, array imperfections, low \acp{snr}, or limited observation intervals, both the estimation performance and the validity of the associated uncertainty measures can deteriorate substantially.

The emergence of deep learning has led to the development of \ac{dnn}-based \ac{doa} estimation methods capable of operating in challenging scenarios where classical array processing algorithms struggle~\cite{al2022review}. Black-box architectures based on \acp{mlp}~\cite{DNN_WITH_Antenna_ARRAY,cong2020robust,feintuch2023neural}, \acp{cnn}~\cite{DOAEstimation_LowSNR,lee2022deep,zheng2024deepdoa}, and attention mechanisms~\cite{lan2023novel,ji2024transmusic} have been proposed to directly map array measurements or covariance-derived features into source directions. An alternative paradigm follows model-based deep learning principles~\cite{shlezinger2023model}, integrating trainable neural components into established array processing pipelines rather than replacing them entirely. Representative examples include deep-learning-aided \ac{music} algorithms~\cite{lee2022ftmr,elbir2020deepmusic}, covariance reconstruction networks~\cite{barthelme2021doa,wu2022gridless,jiang2023toeplitz,shiran2026deep}, and surrogate covariance learning approaches such as SubspaceNet~\cite{shmuel2023subspacenet} and its variants~\cite{DA-MUSIC-2023,xu2024md,gast2025near,zohar2025remote}. By leveraging data-driven processing, these methods can  enable reliable \ac{doa} estimation in regimes where conventional subspace methods often fail. Nevertheless, while modern \acp{dnn} can enhance estimation accuracy, they generally provide only point estimates and do not naturally quantify the confidence associated with their predictions.

Uncertainty quantification in \ac{ai} has  attracted considerable attention in recent years~\cite{gawlikowski2023survey}. Existing approaches include Bayesian neural networks~\cite{jospin2022hands}, which model uncertainty through distributions over network parameters, and were considered for \ac{doa} localization in \cite{fu2026deep}; ensemble-based methods, that approximate predictive distributions using multiple independently trained models~\cite{rahaman2021uncertainty}; and conformal prediction techniques, which provide distribution-free confidence sets~\cite{lindemann2024formal}, and were used for calibrating \acp{dnn} trained for localization in  \cite{khurjekar2023uncertainty, rozenfeld2026uncertainty}. Being model-agnostic methods designed for black-box \acp{dnn}, these methodologies 
are agnostic of the physical structure underlying the estimation problem. More fundamentally, these methods treat uncertainty for black-box \acp{dnn}, rather than of the underlying array processing algorithm. As a result, they do not leverage the  analytical characterization available for classical estimators such as \ac{esprit}, and  their uncertainty measures are not directly tied to the signal model and the recovered spatial covariance. Consequently, despite the significant advances brought by deep learning to \ac{doa} estimation, principled uncertainty quantification that remains reliable under challenging operating conditions is still largely an open problem.

\subsection*{Contributions} 
In this work, we extend model-based deep learning architectures for \ac{doa} estimation, and in particular \ac{dnn}-aided subspace methods based on surrogate covariance recovery, to provide   principled uncertainty quantification in challenging operating conditions. Our approach builds on two complementary observations: $(i)$ \ac{dnn}-aided subspace methods preserve the core processing stages of classical estimators such as \ac{esprit}, while enabling reliable operation under coherent sources, array imperfections, and limited observations~\cite{shmuel2023subspacenet}; and $(ii)$ classical subspace methods are accompanied by rigorous analytical characterizations of their estimation uncertainty~\cite{yuen2002asymptotic,rao2002performance} . By extending these analytical tools and integrating them into a \ac{dnn}-aided subspace processing pipeline, we develop a framework for  jointly recovering \acp{doa} and their associated uncertainty. Our  methodology combines the robustness of modern data-driven learning with the interpretability and theoretical foundations of classical array processing, yielding reliable uncertainty estimates even in scenarios where conventional subspace methods and their corresponding uncertainty analyses break down.

Our main contributions are summarized as follows:
\begin{itemize}
\item {\bf \ac{esprit} Error Covariance Characterization}: We extend existing asymptotic uncertainty analyses of \ac{esprit}~\cite{yuen2002asymptotic} to characterize the {\em full covariance} of its \ac{doa} estimation error. The  formulation enables recovering uncertainty information beyond per-source variance estimates, providing a richer statistical characterization of the estimated directions.

\item {\bf Uncertainty Extraction with SubspaceNet}: We propose an uncertainty-aware \ac{dnn}-aided subspace estimation framework that combines learned surrogate covariance recovery~\cite{shmuel2023subspacenet} with \ac{esprit}-based \ac{doa} estimation and covariance reconstruction. By preserving the subspace processing chain, the proposed architecture retains the interpretability and analytical tractability of classical methods while extending their applicability to challenging  regimes.

\item {\bf Uncertainty-Aware Training}: We develop a dedicated learning methodology that  trains the surrogate covariance recovery module to support both accurate \ac{doa} estimation and reliable uncertainty quantification. This joint design allows the learned representation to both facilitate faithful recovery of both the \acp{doa} and their error covariance.

\item {\bf Extensive Experimentation}: We conduct an extensive numerical study across a broad range of operating conditions, including coherent sources, limited snapshots, low \acp{snr}, and array calibration errors. The results demonstrate that the proposed approach consistently achieves accurate \ac{doa} estimation together with reliable uncertainty characterization, substantially outperforming conventional subspace-based uncertainty quantification methods in challenging scenarios.
\end{itemize}

The rest of this paper is organized as follows: Section~\ref{sec:System Model and Preliminaries} presents the system model and reviews some preliminaries. 
We derive the uncertainty extraction method in Section~\ref{sec: Method} and evaluate it in Section~\ref{sec: Numerical Study}. Section~\ref{sec: Conclusions} provide concluding remarks.  

Throughout this paper, we use
boldface-uppercase for matrices, e.g., $\myVec{X}$, 
boldface-lowercase for vectors, e.g., $\myVec{x}$. 
We denote the $j$th entry of vector $\myVec{x}$ and the $(i,j)$th entry of matrix $\myMat{X}$ by $[\myVec{x}]_j$ and $[\myMat{X}]_{i,j}$, respectively, while ${\rm diag}(\myVec{x})$ is a diagonal matrix with $\myVec{x}$ on its main diagonal. 
We use $(\cdot)^{\mathsf{H}}$, $(\cdot)^{\mathsf{T}}$, $(\cdot)^{\dagger}$, $(\cdot)^{*}$, $\| \cdot\|$, $\Re\{\cdot\}$, $\angle(\cdot)$, and $ \expecteds{\cdot}$  for the hermitian transpose, transpose, Moore–Penrose pseudoinverse, conjugate,  $\ell_2$/Frobenius  norm (for vectors/matrices), real value, phase, and stochastic expectation, respectively, while $\mathbb{C}$ is the set of complex numbers.
%

\section{System Model and Preliminaries}
\label{sec:System Model and Preliminaries} 
In this section, we formulate the  problem of uncertainty-aware \ac{doa} estimation and review the necessary background. We begin in Subsection~\ref{subsec: DoA Estimation Signal Model} by introducing the array signal model. Next, Subsection~\ref{ssec:Problem} formally defines the problem of jointly estimating the source directions and their associated uncertainty. Then, Subsection~\ref{ssec:Preliminaries} reviews the classical and learning-based \ac{doa} estimation methodologies that form the basis of the proposed approach, with an emphasis on \ac{esprit} and SubspaceNet.

\subsection{Signal Model}
\label{subsec: DoA Estimation Signal Model} 
We consider the classical narrowband far-field array processing model, where a set of stationary emitters are observed by a sensor array over multiple temporal snapshots. Specifically, let a \ac{ula} consisting of $N$ sensors with half-wavelength inter-element spacing receive signals from $M$ static sources over $T$ observation intervals. 
Denote by $\myMat{X}\triangleq[\myVec{x}(1),\ldots,\myVec{x}(T)]\in\mathbb{C}^{N\times T}$ the matrix collecting the received array snapshots, where $\myVec{x}(t)\in\mathbb{C}^{N}$ represents the measurement vector acquired at time index $t$. Likewise, let $\myMat{S}\triangleq[\myVec{s}(1),\ldots,\myVec{s}(T)]\in\mathbb{C}^{M\times T}$ collect the transmitted source symbols, where $\myVec{s}(t)\in\mathbb{C}^{M}$ contains the symbols emitted by the $M$ sources at the $t$th snapshot. The unknown directions of arrival are gathered in the vector
$\myVec{\theta}=[\theta_1,\ldots,\theta_M]^{\mathsf T}$.

Under the narrowband propagation model, the received array measurements satisfy 
\begin{equation}
        \myMat{X} = \myMat{A}(\myVec{\theta})\myMat{S} + \myMat{V},
\label{eq:signals_observation_model} 
\end{equation}
where $\myMat{V}$ is the noise matrix with i.i.d. zero-mean entries of variance $\sigma_v^2$, and   $\myMat{A}(\myVec{\theta})\triangleq      \big[ \myVec{a}(\theta_1), \ldots, \myVec{a}(\theta_M) \big] \in \mathbb{C}^{N \times M}$ is the steering matrix whose columns correspond to the array responses associated with the different source directions.  For the considered half-wavelength \ac{ula}, the steering vector associated with an impinging signal arriving from angle $\theta$ is given by  
\begin{equation}
    \myVec{a}(\theta) \triangleq 
    \big[ 1, \, e^{-j\pi \sin(\theta)}, \, \ldots, \, e^{-j\pi (N-1) \sin(\theta)} \big]^{\mathsf{T}}.
    \label{eqn:SteeringVec} 
\end{equation}

\subsection{Problem Formulation}
\label{ssec:Problem}
We aim to design an estimator that maps the array measurements $\myMat{X}$ into both the source directions and a statistical characterization of their estimation error. Specifically, the estimator is required to map $\myMat{X}$ into $\big( \hat{\myVec{\theta}}, \hat{\myMat{\Sigma}} \big)$, where $\hat{\myVec{\theta}}\in\mathbb{R}^{M}$ is the estimated \ac{doa} vector, and $\hat{\myMat{\Sigma}}\in\mathbb{R}^{M\times M}$ is an estimate of the covariance matrix of the corresponding \ac{doa} estimation error. 
By defining
\begin{equation}
    \Delta \myVec{\theta} \triangleq \hat{\myVec{\theta}}-\myVec{\theta},
    \label{eqn:DeltaTheta}
\end{equation}
the accuracy of the \ac{doa} estimates is evaluated in terms of the \ac{mse} $\mathbb{E}\big[\big\|\Delta \myVec{\theta}  \big\|^{2} \big]$,  while $\hat{\myMat{\Sigma}}$ should approximate the error covariance (assuming unbiased estimates)
\begin{equation}
\myMat{\Sigma} \triangleq \mathbb{E}\left[ \Delta \myVec{\theta} \Delta \myVec{\theta} ^{\mathsf T} \right].
\label{eqn:EstCov}
\end{equation}
Unlike uncertainty representations based only on the marginal error variances, the matrix $\hat{\myMat{\Sigma}}$ should capture both the reliability of each individual \ac{doa} estimate through its diagonal entries and the statistical coupling between errors associated with different source directions through its off-diagonal entries.

We focus on operating conditions in which conventional subspace methods are known to be sensitive to model mismatch and finite-sample effects. In particular, the estimator should operate reliably under the following challenges:
\begin{enumerate}[label={\em C\arabic*}]
    \item \label{itm:Coherent} \emph{Coherent sources}, for which the  covariance matrix of the source signal $\myVec{s}(t)$ is rank-deficient.
    \item \label{itm:Miscalibration} \emph{Array miscalibration}, where the true array response deviates from the nominal steering model in \eqref{eqn:SteeringVec}.
    \item \label{itm:Snapshots} \emph{Limited observations}, corresponding to a small number of temporal snapshots ($T$), low \ac{snr} (large $\sigma_v^2$), or both.
\end{enumerate}
These scenarios are especially challenging for uncertainty-aware \ac{doa} estimation, since the same deviations that degrade the accuracy of classical subspace methods (briefly recalled in the following subsection) also invalidate the statistical assumptions underlying their uncertainty analyses.

To cope with these challenges, we adopt a data-aided design approach. During training, we assume access to a labeled dataset containing array measurements and their corresponding ground-truth \acp{doa}, given by 
\begin{equation}
    \label{eqn:DataSet}
    \mySet{D} = \big\{ \myMat{X}^{(i)}, \myVec{\theta}^{(i)} \big\}_{i=1}^{|\mySet{D}|}. 
\end{equation}
The learned estimator is then deployed to infer both $\hat{\myVec{\theta}}$ and $\hat{\myMat{\Sigma}}$ from previously unseen array measurements. Our goal is therefore to exploit the training data to enhance robustness to the challenges \ref{itm:Coherent}-\ref{itm:Snapshots}, while retaining the interpretable uncertainty characterization enabled by subspace-based \ac{doa} recovery.
 
\subsection{Preliminaries}
\label{ssec:Preliminaries} 

\subsubsection{Subspace Methods}
A leading class of \ac{doa} estimators is based on the eigenspace structure of the array covariance matrix. Let $\myMat{R}_S\triangleq \mathbb{E}[\myVec{s}(t)\myVec{s}^{\mathsf H}(t)]$ denote the source covariance matrix. It follows from \eqref{eq:signals_observation_model} that the covariance of the received array snapshots is given by   
\begin{equation}
    \myMat{R}_X = \mathbb{E}\!\left[\myVec{x}(t)\myVec{x}^\mathsf{H}(t)\right]
    = \myMat{A}(\myVec{\theta}) \,\myMat{R}_S\, \myMat{A}^\mathsf{H}(\myVec{\theta}) + \sigma_v^2 \myMat{I}_{N}.
    \label{eq:covariance_matrix}
\end{equation}
When $\myMat{R}_S$ is full rank, the eigenspace of $\myMat{R}_X$ can be decomposed into a signal subspace $\myMat{E}_{\rm S}$, spanned by the steering vectors, and a noise subspace $\myMat{E}_{\rm N}$, which is orthogonal to them. This orthogonality implies that $
    \| \myMat{E}_{\rm N}^\mathsf{H}\,\myVec{a}({\theta}_i) \|^2 = 0$, 
and forms the basis of subspace-based \ac{doa} estimation~\cite{pillai2012array}. In practice, the covariance matrix is replaced by its empirical estimate
$\hat{\myMat{R}}_X=\frac{1}{T}\myMat{X}\myMat{X}^{\mathsf H}$, and the \acp{doa} are inferred using the nominal steering model in \eqref{eqn:SteeringVec}. Consequently, subspace methods are sensitive to \ref{itm:Coherent}-\ref{itm:Snapshots}: coherent sources violate the full-rank assumption on $\myMat{R}_S$, limited snapshots degrade the accuracy of $\hat{\myMat{R}}_X$, and array miscalibration induces mismatch in the steering model.
     
A representative subspace method is {\em \ac{esprit}}~\cite{roy1989esprit}, which exploits the rotational invariance of the signal subspace. For a \ac{ula}, the array can be viewed as two subarrays whose steering matrices, denoted by $\myMat{A}_1(\myVec{\theta})$ and $\myMat{A}_2(\myVec{\theta})$, satisfy  
$\mathbf{A}_2(\boldsymbol{\theta}) = \mathbf{A}_1(\boldsymbol{\theta}) \myMat{\Lambda}_{\myVec{\theta}}$,
where 
\begin{equation}
\label{eqn:lambdaDef}
\myMat{\Lambda}_{\myVec{\theta}} \triangleq \mathrm{diag}\!\left(e^{-j\pi\sin(\theta_1)}, \ldots, e^{-j\pi\sin(\theta_M)}\right). 
\end{equation}
Let $\myMat{J}_1,\myMat{J}_2\in\{0,1\}^{(N-k)\times N}$ be the  subarray selection matrices, where $k$ is the displacement between the subarrays. The selection should hold the rotational invariance property, namely $\myMat{J}_2 \myVec{a}(\theta)\equiv e^{-j \pi k \sin(\theta)}\myMat{J}_1 \myVec{a}(\theta)$, most commonly obtained by setting $\myMat{J}_1 = [\myMat{I}_{N-k},  \myMat{0}]$ and $\myMat{J}_2 = [\myMat{0}, \myMat{I}_{N-k}]$.
The selected signal subspaces are then
$\myMat{E}_x=\myMat{J}_1\myMat{E}_{\rm S}$ and
$\myMat{E}_y=\myMat{J}_2\myMat{E}_{\rm S}$.

The rotational invariance property implies that
$\myMat{E}_y = \myMat{E}_x \,\myMat{F}$,
for a nonsingular transformation matrix $\myMat{F} \in \mathbb{C}^{M\times M}$,  obtained as
\begin{align}
    \myMat{F}     =     \myMat{E}_{x}^{\dagger}     \myMat{E}_{y}     =     \left(     \myMat{J}_{1}\myMat{E}_{\rm S}    \right)^{\dagger}     \myMat{J}_{2}\myMat{E}_{\rm S}.     \label{eqn:esprit_psi_matrix}
\end{align}  
The eigenvalues of $\myMat{F}$, denoted $\myVec{\lambda}= [\lambda_1,\ldots,\lambda_M]^{\mathsf{T}}$, are the diagonal entries of $\myMat{\Lambda}_{\myVec{\theta}}$, and the \acp{doa} are recovered as 
\begin{equation}
\label{eqn:ESPRIT}
\hat{\myMat{\theta}} 
= \mathrm{ESPRIT}\left(\hat{\myMat{R}}_X\right) 
= \sin^{-1} \left(- {\angle\,\myVec{\lambda}}/{\pi} \right). 
\end{equation} 

\subsubsection{\ac{dnn}-Based \ac{doa} Recovery}
Deep learning methods have been shown to improve \ac{doa} estimation in regimes where conventional subspace methods are sensitive to \ref{itm:Coherent}-\ref{itm:Snapshots}. A direct approach is to employ an {\em end-to-end} \ac{dnn} that maps the received measurements, or features extracted from them, into the estimated source directions. For instance, the input to the network can be the snapshot matrix $\myMat{X}$ itself or a covariance-based feature such as $\frac{1}{T}\myMat{X}\myMat{X}^{\mathsf H}$~\cite{DNN_WITH_Antenna_ARRAY,DOAEstimation_LowSNR}. Letting $\myVec{\psi}$ denote the trainable parameters and ${\rm DNN}(\cdot;\myVec{\psi})$ be the induced mapping, such estimators take the form
\begin{equation}
\label{eqn:E2E}
    \hat{\myVec{\theta}}^{\rm E2E} = {\rm DNN}(\myMat{X}; \myVec{\psi}).
\end{equation}

An alternative approach follows model-based deep learning, where trainable components are incorporated into classical signal processing pipelines rather than replacing them entirely~\cite{shlezinger2020model}. In particular, {\em SubspaceNet}~\cite{shmuel2023subspacenet} learns a surrogate covariance matrix $\tilde{\myMat{R}}_X$ from the measurements $\myMat{X}$, while retaining the downstream subspace-based \ac{doa} estimator. This can be viewed as a discriminative architecture whose trainable component is optimized through the final \ac{doa} recovery objective~\cite{shlezinger2022discriminative}. When \ac{esprit} is used as the subspace estimator, the resulting estimate is 
\begin{equation}
\label{eqn:SSN}
       \hat{\myVec{\theta}}^{\rm SSN} = {\rm ESPRIT}\left(\tilde{\myMat{R}}_X={\rm DNN}(\myMat{X}; \myVec{\psi})\right).
\end{equation}
Thus, the learned covariance is trained to support the subsequent subspace operation, enabling \ac{esprit}-based \ac{doa} recovery under \ref{itm:Coherent}-\ref{itm:Snapshots} while preserving the interpretable processing structure of classical subspace methods. Still, only point-wise estimates of $\myVec{\theta}$ are provided, without providing uncertainty measures as in \eqref{eqn:EstCov}.


\section{AI-Aided Joint DoA and  Covariance Recovery}
\label{sec: Method} 
We now present the proposed framework for jointly estimating the \acp{doa} and characterizing their error covariance. The key idea is to leverage the complementary advantages of model-based deep learning and classical subspace analysis: \ac{dnn}-aided surrogate covariance recovery, as in SubspaceNet, enables subspace-based \ac{doa} estimation in challenging regimes under \ref{itm:Coherent}-\ref{itm:Snapshots}. While this is supported by both black-box \eqref{eqn:E2E} and model-based learning \eqref{eqn:SSN}, they substantially differ in their ability to provide uncertainty. By retaining the \ac{esprit} processing  as in \eqref{eqn:SSN}, the resulting estimator preserves the analytical mechanisms that allow uncertainty characterization through perturbation analysis. This combination allows us characterize uncertainty under \ref{itm:Coherent}-\ref{itm:Snapshots}, while providing not only accurate point estimates but also interpretable error covariance estimates.

Our design is motivated by the observation that the same subspace quantities used by \ac{esprit} to recover the source directions also determine the sensitivity of the recovered angles to perturbations in the covariance matrix. Classical asymptotic analyses of \ac{esprit}~\cite{yuen2002asymptotic} exploit this relation to characterize the variance of individual \ac{doa} estimates under idealized conditions. Here, we first extend this characterization to recover the full covariance matrix of the \ac{doa} estimation error in Subsection~\ref{subsec:esprit_uncertainty_characterization}, providing a full characterization that captures both the marginal uncertainty of each estimated direction and the correlations between the estimation errors of different sources. This full covariance characterization is then incorporated into the SubspaceNet pipeline in Subsection~\ref{ssec:Algorithm}, which  we adapt to have its   learned surrogate covariance optimized for both angle recovery and uncertainty estimation. We conclude with a discussion in Subsection~\ref{ssec:Discssion}.

\subsection{ESPRIT Uncertainty Characterization}
\label{subsec:esprit_uncertainty_characterization}
We begin by characterizing the uncertainty associated with \ac{esprit}-based \ac{doa} recovery. When operating under its conventional assumptions, namely calibrated arrays, non-coherent sources, and sufficiently many snapshots, the estimation error of \ac{esprit} can be analyzed using asymptotic perturbation tools~\cite{yuen2002asymptotic}. Such analyses relate the fluctuations of the empirical covariance matrix to the resulting perturbations in the signal subspace, the rotational invariance matrix, and eventually the recovered \acp{doa}. 

However, existing \ac{esprit} uncertainty characterizations are not directly suited to our setting. First, they are typically derived for specific \ac{esprit} variants based on two non-overlapping subarrays, obtained by partitioning the $N$ sensors into two disjoint subarrays of size $N/2$. Practical \ac{esprit} implementations often rely on overlapping subarrays, and in particular on maximally overlapping subarrays, due to their improved effective aperture and statistical efficiency. Second, existing characterizations usually focus on the  \ac{mse} of each recovered direction, and do not capture the full statistical coupling between the errors of different \ac{doa} estimates. This coupling can be important in downstream tasks such as tracking and sensor fusion,, where the joint distribution of the estimated parameters affects subsequent inference.
To address these limitations, we next extend the perturbation analysis of \ac{esprit} in~\cite{yuen2002asymptotic} to overlapping-subarray \ac{esprit}, and derive the complete covariance matrix of the \ac{doa} estimation error, rather than only its diagonal entries.

\subsubsection{Definitions}
To formulate our covariance characterization, we  consider \ac{esprit} applied to measurements taken from
a calibrated \ac{ula}.
The selection matrices $\myMat{J}_{1}, \myMat{J}_{2}$ can represent any subarray division for which rotational invariance holds.
The input covariance $\hat{\myMat{R}}_X$ processed by \ac{esprit} is an estimate of the true covariance \eqref{eq:covariance_matrix}, and we define the deviation as 
\begin{align}
    \Delta \myMat{R} \triangleq \hat{\myMat{R}}_X
    -
    \myMat{R}_X.
    \label{eqn:covariance_perturbation}
\end{align}


Next, for the matrix $\myMat{F}$ computed by \ac{esprit} in \eqref{eqn:esprit_psi_matrix}, we use  \(\myVec{v}_i\) and \(\myVec{q}_i\) to respectively denote its right and left eigenvectors  associated with the eigenvalue \(\lambda_i\). Furthermore, let $(\alpha_i, \myVec{\rho}_i)$ be the $i$th eigenvalue and eigenvector  of true covariance $\myMat{R}_X$ in \eqref{eq:covariance_matrix}. For each $g,h\in \{1,\ldots,M\}$, we define the $N \times N$ matrices
\begin{subequations}
    \label{eq:cov_moments}
\begin{align} 
\TSigH_{gh}
&\triangleq
\sum_{\substack{l=1\\ l\neq g}}^{N}
\sum_{\substack{n=1\\ n\neq h}}^{N}
\sum_{a_1,a_2=1}^{N}
\sum_{b_1,b_2=1}^{N}
 [\myVec{\rho}^{*}_{l}]_{a_1}\, [\myVec{\rho}_g]_{a_2}\, [\myVec{\rho}_n]_{b_1}\, [\myVec{\rho}^{*}_h]_{b_2}
\notag \\
&\quad\times
\frac{\mathbb{E}\!\left[[\Delta \myMat{R}]_{a_1 a_2}\,
                  [\Delta \myMat{R}]^{*}_{b_1 b_2}\right]}{(\alpha_g-\alpha_l)(\alpha_h-\alpha_n)}
\myVec{\rho}_l \myVec{\rho}_n^{H},
\label{eq:cov_moments_hermit_eq} 
\end{align}
and 
\begin{align} 
\TSigT_{gh}
&\triangleq
\sum_{\substack{l=1\\ l\neq g}}^{N}
\sum_{\substack{n=1\\ n\neq h}}^{N}
\sum_{a_1,a_2=1}^{N}
\sum_{b_1,b_2=1}^{N}
 [\myVec{\rho}^{*}_{l}]_{a_1}\, [\myVec{\rho}_g]_{a_2}\, [\myVec{\rho}_n]_{b_1}\, [\myVec{\rho}_h]_{b_2}
\notag \\
&\quad\times
\frac{\mathbb{E}\!\left[[\Delta \myMat{R}]_{a_1 a_2}\,
                  [\Delta \myMat{R}]_{b_1 b_2}\right]}{(\alpha_g-\alpha_l)(\alpha_h-\alpha_n)}
\myVec{\rho}_l \myVec{\rho}_n^{T}.
\label{eq:cov_moments_transpose_eq} 
\end{align}
\end{subequations}
We use \eqref{eq:cov_moments} to define for  each $i,j\in \{1,\ldots, M\}$, the matrices
\begin{align*}
    \SigH_{ij}
    \!\!\triangleq\!
    \sum_{g=1}^{M}
    \sum_{h=1}^{M}
    [\myVec{v}_i]_g
    [\myVec{v}_j]_{h}^{*}
    \TSigH_{gh},
    \,
    \SigT_{ij}
    \!\!\triangleq\!
    \sum_{g=1}^{M}
    \sum_{h=1}^{M}
    [\myVec{v}_i]_g
    [\myVec{v}_j]_{h}
    \TSigT_{gh}. 
\end{align*}

Next, by writing $\myMat{W}_i\triangleq \myMat{J}_{1} -
    \lambda_i^{*}\myMat{J}_{2}$ and $\myMat{U}_i\triangleq
    \myMat{J}_{2}     -     \lambda_i\myMat{J}_{1}$ for each $i\in \{1,\ldots,M\}$, we define the $M\times M$ matrices $\myMat{K}_{\lambda}$ and $\widetilde{\myMat{K}}_{\lambda}$ whose $(i,j)$th elements are obtained as
    \begin{subequations}
        \label{eq:Klambda_pairwise_full}    
\begin{align}
    \left[
    \myMat{K}_{\lambda}
    \right]_{i,j}
    &=
    \myVec{q}_i^{\mathsf{H}}
    \myMat{E}_x^{\dagger}
    \myMat{W}_i
    \SigH_{ij}
    \myMat{W}_j^{\mathsf{H}}
    \left(
    \myMat{E}_x^{\dagger}
    \right)^{\mathsf{H}}
    \myVec{q}_j,
    \label{eq:Klambda_pairwise_full_cov}
    \\
    \big[
    \widetilde{\myMat{K}}_{\lambda}
    \big]_{i,j}
    &=
    \myVec{q}_i^{\mathsf{H}}
    \myMat{E}_x^{\dagger}
    \myMat{U}_i
    \SigT_{ij}
    \myMat{U}_j^{\mathsf{T}}
    \left(
    \myMat{E}_x^{\dagger}
    \right)^{\mathsf{T}}
    \myVec{q}_j^{*}.
    \label{eq:Ktilde_lambda_pairwise_full_cov}
\end{align}
\end{subequations}
We use \eqref{eq:Klambda_pairwise_full} along with the  matrix $\myMat{\Lambda}_{\myVec{\theta}}$ defined in \eqref{eqn:lambdaDef} to form 
\begin{align}
    \myMat{C}_{{\myVec{\theta}}}  &\triangleq     \myMat{\Lambda}_{{\myVec{\theta}}}^{-1}     \myMat{K}_{\lambda}     \myMat{\Lambda}_{{\myVec{\theta}}}^{-\mathsf{H}} -   \myMat{\Lambda}_{{\myVec{\theta}}}^{-1}     \widetilde{\myMat{K}}_{\lambda}   \myMat{\Lambda}_{{\myVec{\theta}}}^{-\mathsf{T}} \in \mathbb{C}^{M\times M}.
    \label{eqn:propagated_eigenvalue_covariance}
\end{align}
Finally, we define the real-valued diagonal matrix
\begin{equation}
     \myMat{D}_{\myVec{\theta}}  \triangleq   \operatorname{diag}    \left( \frac{1}{\pi\cos\theta_1},    \ldots,    \frac{1}{\pi\cos\theta_M}    \right).
     \label{eqn:DthetaDef}
\end{equation}

\subsubsection{Covariance Characterization}
Using the quantities defined above, we can now characterize the error covariance \eqref{eqn:EstCov} of \ac{esprit} in the asymptotic regime of sufficiently small $\Delta\myMat{R}$. This is stated in the following theorem:
\begin{theorem}
\label{thm:full_doa_covariance_overlapping_esprit}
    Consider \ac{esprit} applied to recover \acp{doa} $\{\theta_i\}_{i=1}^M$. Then, when $\myMat{F}$ in \eqref{eqn:esprit_psi_matrix} exists and all its eigenvalues $\{\lambda_i\}_{i=1}^M$ are of order one, and when $\|\Delta \myMat{R}\| \ll 1$ with probability one, it holds that for all $\myVec{\theta}$ for which $\cos \theta_i$ bounded away from zero for all $i\in \{1,\ldots, M\}$, the error covariance satisfies
    \begin{equation}          
    \myMat{\Sigma}  = \frac{1}{2}     \myMat{D}_{\myVec{\theta}}
    \Re     \left\{     \myMat{C}_{{\myVec{\theta}}}
    \right\}     \myMat{D}_{\myVec{\theta}} +  \mathcal{O} \left( \mathbb{E}[\left\|  \Delta \myMat{R}  \right\|^{3}]  \right).
\label{eqn:full_doa_covariance_overlapping_esprit}
    \end{equation} 
\end{theorem}

\begin{IEEEproof}
    See Appendix~\ref{app:proof_full_doa_covariance}.
\end{IEEEproof}

Theorem~\ref{thm:full_doa_covariance_overlapping_esprit}
 provides a characterization of the full covariance (including inter source correlation).
 This theorem specializes the uncertainty formulation provided in \cite{yuen2002asymptotic}, as stated in the following corollary: 
 \begin{corollary}[Variance-only non-overlapping ESPRIT]
\label{cor:nonoverlap_variance_special_case}
For selection matrices $\myMat{J}_1  = [\myMat{I}_{N/2}, \myMat{0}] , \myMat{J}_2= [\myMat{0}, \myMat{I}_{N/2}]$ and assuming that
 $|\lambda_i| = 1$, it holds that \eqref{eq:DOA_var_half} coincides with \cite[Eq. 58]{yuen2002asymptotic}.
\end{corollary}

\begin{IEEEproof}
See Appendix~\ref{app:nonoverlap_variance_special_case}.  
\end{IEEEproof}

\subsubsection{Implications}
Theorem~\ref{thm:full_doa_covariance_overlapping_esprit} shows that an estimated covariance matrix that satisfies the stated assumptions can be used for both  \ac{doa} recovery and uncertainty quantification. The same recovered covariance matrix $\hat{\myMat{R}}_X$  is first used to recover
\acp{doa} as  $\hat{\myVec{\theta}}$ through \ac{esprit}, and can then be used  to estimate the full uncertainty covariance $\hat{\myMat{\Sigma}}$ using \eqref{eqn:full_doa_covariance_overlapping_esprit}. 
The diagonal entries of \(\hat{\myMat{\Sigma}}\) quantify the marginal uncertainty
of each estimated \ac{doa}, while the off-diagonal entries quantify the coupling between
the estimation errors of different sources. 

\subsection{Uncertainty-Aware SubspaceNet}
\label{ssec:Algorithm}
We next use Theorem~\ref{thm:full_doa_covariance_overlapping_esprit} to extend SubspaceNet into an uncertainty-aware \ac{doa} estimator. The proposed architecture retains the surrogate covariance recovery principle of SubspaceNet, but uses the recovered covariance for two coupled model-based operations: \ac{esprit}-based \ac{doa} estimation and perturbation-based reconstruction of the corresponding error covariance. The method is not tied to a particular choice of subarrays. The selection matrices $\myMat{J}_1,\myMat{J}_2$ are treated as design parameters satisfying the rotational invariance property required by \ac{esprit}.

Let $g_{\myVec{\psi}}(\cdot)$ denote the trainable surrogate covariance recovery mapping, including any Hermitian or positive-semidefinite projection used by the implemented SubspaceNet architecture (i.e., the \ac{dnn} in \eqref{eqn:SSN}). Given the received measurements $\myMat{X}$, we write
\begin{equation}
\tilde{\myMat{R}}_X = g_{\myVec{\psi}}(\myMat{X}).
\label{eqn:surrogate_covariance_uassn}
\end{equation}
The proposed estimator maps $\myMat{X}$ into
\begin{equation}
\left( \hat{\myVec{\theta}}, \hat{\myMat{\Sigma}} \right)
= {\rm USSN}_{\myMat{J}_1,\myMat{J}_2} \left( \myMat{X};\myVec{\psi}\right),
\label{eqn:ussn_mapping}
\end{equation}
where $\hat{\myVec{\theta}}$ is obtained by applying \ac{esprit} to $\tilde{\myMat{R}}_X$, and $\hat{\myMat{\Sigma}}$ is obtained by substituting the resulting \ac{esprit} quantities into the covariance characterization of Theorem~\ref{thm:full_doa_covariance_overlapping_esprit}.

\subsubsection{Inference}
During inference, the surrogate covariance $\tilde{\myMat{R}}_X$ replaces the empirical covariance processed by \ac{esprit}. Specifically, we compute
\begin{equation}
\hat{\myVec{\theta}} =
{\rm ESPRIT}_{\myMat{J}_1,\myMat{J}_2} \left( \tilde{\myMat{R}}_X \right),
\label{eqn:ussn_doa_recovery}
\end{equation}
where ${\rm ESPRIT}_{\myMat{J}_1,\myMat{J}_2}(\cdot)$ denotes the \ac{esprit} procedure in \eqref{eqn:esprit_psi_matrix}-\eqref{eqn:ESPRIT} using the selected subarrays. This computation also provides the auxiliary quantities required for uncertainty extraction, including the estimated signal subspace $\tilde{\myMat{E}}_{\rm S}$, the selected subspaces
\begin{equation*}
\tilde{\myMat{E}}_x = \myMat{J}_1\tilde{\myMat{E}}_{\rm S},
\qquad
\tilde{\myMat{E}}_y = \myMat{J}_2\tilde{\myMat{E}}_{\rm S},
\end{equation*}
the estimated rotational invariance matrix
$\tilde{\myMat{F}} = \tilde{\myMat{E}}_x^{\dagger} \tilde{\myMat{E}}_y$, 
and its eigenvalues and right/left eigenvectors
${(\hat{\lambda}_i,\hat{\myVec{v}}_i,\hat{\myVec{q}}_i)}_{i=1}^{M}$.

The covariance characterization in Theorem~\ref{thm:full_doa_covariance_overlapping_esprit} depends on the second-order moments of the covariance perturbation $\Delta\myMat{R}$. These statistics are not available during inference. We therefore adopt a plug-in approximation in which the unknown covariance $\myMat{R}_X$ is replaced by the learned surrogate covariance $\tilde{\myMat{R}}_X$. Assuming locally that the snapshots are zero-mean complex Gaussian with covariance $\tilde{\myMat{R}}_X$, Wick's factorization yields
\begin{subequations}
\label{eqn:gaussian_plugin_moments}
\begin{align}
\widehat{C}^{(\mathsf H)}_{a_1a_2,b_1b_2} &\triangleq \frac{1}{T}\left[ \tilde{\myMat{R}}_X \right]_{a_1b_1} \left[ \tilde{\myMat{R}}_X \right]_{b_2a_2},
\label{eqn:gaussian_plugin_moments_H}
\\
\widehat{C}^{(\mathsf T)}_{a_1a_2,b_1b_2} &\triangleq \frac{1}{T} \left[ \tilde{\myMat{R}}_X \right]_{a_1b_2} \left[ \tilde{\myMat{R}}_X \right]_{b_1a_2}.
\label{eqn:gaussian_plugin_moments_T}
\end{align}
\end{subequations}
These plug-in moments are substituted into \eqref{eq:cov_moments} in place of the corresponding perturbation moments. The eigenvalue decomposition of $\tilde{\myMat{R}}_X$ provides the quantities $\{(\hat{\alpha}_i, \hat{\myVec{\rho}}_i)\}_{i=1}^{N}$ used in \eqref{eq:cov_moments}. Then, replacing all quantities in \eqref{eq:cov_moments}-\eqref{eqn:DthetaDef} by their estimates yields $\widehat{\myMat{C}}_{\hat{\myVec{\theta}}}$ and $\myMat{D}_{\hat{\myVec{\theta}}}$. The predicted \ac{doa} error covariance is finally computed as
\begin{equation}
    \hat{\myMat{\Sigma}} =  \frac{1}{2}     \myMat{D}_{\hat{\myVec{\theta}}}     \Re\!\left\{    \widehat{\myMat{C}}_{\hat{\myVec{\theta}}} \right\}  \myMat{D}_{\hat{\myVec{\theta}}}.
    \label{eqn:ussn_covariance_output}
\end{equation}
The diagonal entries of $\hat{\myMat{\Sigma}}$ quantify the marginal uncertainty of the recovered \acp{doa}, while its off-diagonal entries quantify the correlations between their estimation errors. The overall procedure is summarized as Algorithm~\ref{alg:ussn_inference}.

\begin{algorithm}
\caption{Uncertainty-Aware SubspaceNet Inference}
\label{alg:ussn_inference}
\SetKwInOut{Initialization}{Init}
\Initialization{Trained \ac{dnn}  $\myVec{\psi}$;  selection matrices $\myMat{J}_1,\myMat{J}_2$}
\SetKwInOut{Input}{Input}
\Input{Array measurements $\myMat{X}$; number of sources $M$;}

Compute the surrogate covariance
$\tilde{\myMat{R}}_X = g_{\myVec{\psi}}(\myMat{X})$;

Compute the signal subspace $\tilde{\myMat{E}}_{\rm S}$ from the \ac{evd} of $\tilde{\myMat{R}}_X$;

Set
$\tilde{\myMat{E}}_x = \myMat{J}_1\tilde{\myMat{E}}_{\rm S}$ and
$\tilde{\myMat{E}}_y = \myMat{J}_2\tilde{\myMat{E}}_{\rm S}$;

Compute the \ac{esprit} rotational matrix
$\tilde{\myMat{F}} = \tilde{\myMat{E}}_x^{\dagger}\tilde{\myMat{E}}_y$;

Set ${\hat{\lambda}_i,\hat{\myVec{v}}_i,\hat{\myVec{q}}_i}_{i=1}^{M}$ from the \ac{evd} of $\tilde{\myMat{F}}$;

Recover the \acp{doa} $\hat{\myVec{\theta}} =
\sin^{-1}\big(-\angle\hat{\myVec{\lambda}}/\pi\big)$;

Set ${(\hat{\alpha}_i,\hat{\myVec{\rho}}_i)}_{i=1}^{N}$ from the \ac{evd} of $\tilde{\myMat{R}}_X$;

Estimate  
$\widehat{C}^{(\mathsf H)}_{a_1a_2,b_1b_2}$ and
$\widehat{C}^{(\mathsf T)}_{a_1a_2,b_1b_2}$
using  \eqref{eqn:gaussian_plugin_moments};

Estimate 
$\TSigH_{gh}$ and $\TSigT_{gh}$
 using  \eqref{eq:cov_moments};

Compute  
${\myMat{K}}_{\lambda}$ and
${\widetilde{\myMat{K}}}_{\lambda}$
using \eqref{eq:Klambda_pairwise_full};

Form $\widehat{\myMat{C}}_{\hat{\myVec{\theta}}}$ using
\eqref{eqn:propagated_eigenvalue_covariance};

Compute the \ac{doa} error covariance estimate
$\hat{\myMat{\Sigma}}$
using \eqref{eqn:ussn_covariance_output};

\KwRet{$\hat{\myVec{\theta}},\hat{\myMat{\Sigma}}$}
\end{algorithm}

\subsubsection{Source Matching and Losses}
Training is carried out using the labeled dataset $\mySet{D}$ in \eqref{eqn:DataSet}. Since \ac{esprit} returns an unordered set of directions, the estimated sources must be matched to the ground-truth sources before evaluating either the \ac{doa} error or the covariance prediction. Let $\mathcal{P}_M$ be the set of $M\times M$ permutation matrices. For a training pair $(\myMat{X},\myVec{\theta})$, we define the optimal assignment as
\begin{equation}
    \myMat{P}^{\star} = \arg\min_{\myMat{P}\in\mathcal{P}_M}     \left\| {\rm mod}_{\pi}\left( \myVec{\theta} - \myMat{P}\hat{\myVec{\theta}} \right) \right\|^2,
    \label{eqn:optimal_assignment_ussn}
\end{equation}
where ${\rm mod}_{\pi}(\cdot)$ denotes the periodic angular error over the considered \ac{doa} domain. The matched angular error vector is
\begin{equation}
    \myVec{e}_{\theta} = {\rm mod}_{\pi} \left( \myVec{\theta} - \myMat{P}^{\star}\hat{\myVec{\theta}} \right).
    \label{eqn:matched_error_vector}
\end{equation}
The corresponding covariance matrix must be permuted using the same assignment $
    \bar{\myMat{\Sigma}} =  \myMat{P}^{\star} \hat{\myMat{\Sigma}}  \left(  \myMat{P}^{\star}  \right)^{\mathsf T}$. 
This step is needed for training with the full covariance matrix, since both its diagonal and off-diagonal entries are source-order dependent.

The \ac{doa} estimation loss is defined as the \ac{rmspe}
\begin{equation}
    \ell_{\rm DoA}  \left(  \myMat{X},\myVec{\theta};\myVec{\psi}  \right) = \left( \frac{1}{M} \left\|  \myVec{e}_{\theta} \right\|^2 \right)^{1/2}.
    \label{eqn:doa_loss_ussn}
\end{equation}
To train the model to produce reliable covariance estimates, we use a loss  that compares the estimated covariance with the empirical error covariance. Since \eqref{eqn:ussn_covariance_output} is obtained through an asymptotic approximation and may not be strictly positive definite numerically, we use the regularized covariance
\begin{equation}
    \bar{\myMat{\Sigma}}_{\epsilon} = \Pi_{\mathbb{S}_{+}^{M}}  \left( \frac{1}{2} \big( \bar{\myMat{\Sigma}}  + \bar{\myMat{\Sigma}}^{\mathsf T}  \big)   \right)
    +  \epsilon \cdot \myMat{I}_{M},
    \label{eqn:regularized_covariance}
\end{equation}
where $\Pi_{\mathbb{S}_{+}^{M}}(\cdot)$ denotes projection onto the cone of positive semidefinite matrices and $\epsilon>0$ is a small diagonal loading coefficient. The uncertainty loss is then set to
\begin{equation}
    \ell_{\rm UQ} \left( \myMat{X},\myVec{\theta};\myVec{\psi} \right)
    = \left[
        \frac{1}{M} \sum_{i=1}^{M}
        \left(
            \left[\bar{\myMat{\Sigma}}\right]_{i,i}
            - \left[\myVec{e}_{\theta}\right]_{i}^{2}
        \right)^{2}
    \right]^{1/2}.
    \label{eqn:uq_loss_ussn}
\end{equation}
This objective encourages the predicted covariance to explain the observed \ac{doa} estimation error. 
\subsubsection{Training Procedure}
We train the surrogate covariance recovery mapping $g_{\myVec{\psi}}$ in two stages. The first stage follows the original SubspaceNet training principle and optimizes only the \ac{doa} estimation accuracy. Its objective is
\begin{equation}
    \mathcal{L}^{(1)}_{\mySet{D}} \left( \myVec{\psi} \right)  =
    \frac{1}{|\mySet{D}|} \sum_{(\myMat{X},\myVec{\theta})\in\mySet{D}} \ell_{\rm DoA} \left(  \myMat{X},\myVec{\theta};\myVec{\psi} \right).
    \label{eqn:stage1_loss_ussn}
\end{equation}
This stage initializes the learned covariance to support reliable \ac{esprit}-based angle recovery.

The second stage fine-tunes the same surrogate covariance recovery module using both the \ac{doa} accuracy loss and the uncertainty loss. The dataset-level objective is
\begin{equation}
    \mathcal{L}^{(2)}_{\mySet{D}}\! \left( \myVec{\psi} \right) \!= \! \frac{1}{|\mySet{D}|} \!\sum_{(\myMat{X},\myVec{\theta})\in\mySet{D}} \!\! \ell_{\rm DoA} \!\left( \myMat{X},\myVec{\theta};\myVec{\psi} \right)\! + \!\mu \ell_{\rm UQ}\! \left( \myMat{X},\myVec{\theta};\myVec{\psi} \right),
    \label{eqn:stage2_loss_ussn}
\end{equation}
where $\mu\geq 0$ controls the relative weight assigned to uncertainty calibration. It can be fixed or gradually increased during training to avoid destabilizing the initial \ac{doa}-oriented solution.

The training procedure based on \ac{sgd} iterations is summarized in Algorithm~\ref{alg:ussn_training}. The \ac{dnn} is trained to output a surrogate covariance matrix whose induced \ac{esprit} estimate and perturbation-based covariance reconstruction are jointly reliable. Thus, uncertainty quantification is obtained through the preserved subspace structure, rather than through an additional black-box uncertainty head.

\begin{algorithm}
\caption{Training Uncertainty-Aware SubspaceNet}
\label{alg:ussn_training}
\SetKwInOut{Initialization}{Init}
\Initialization{Learning rate $\eta$; epochs $(e_1,e_2)$; $\#$ batches $B$; \\ hyperparameters $\mu, \epsilon$;  matrices $\myMat{J}_1,\myMat{J}_2$; initial  $\myVec{\psi}$}
\SetKwInOut{Input}{Input}
\Input{Training dataset $\mySet{D}$}

\nonl\textbf{Stage 1: DoA-Oriented Training};

\For{$e=1,\ldots,e_1$}
{
Randomly divide $\mySet{D}$ into $B$ batches ${\mySet{D}^{(b)}}_{b=1}^{B}$;

\For{each batch $\mySet{D}^{(b)}$}
{
    Compute $\tilde{\myMat{R}}_X = g_{\myVec{\psi}}(\myMat{X})$ for all $(\myMat{X},\myVec{\theta})\in\mySet{D}^{(b)}$\;

    Estimate $\hat{\myVec{\theta}}$ using \ac{esprit} with  $\myMat{J}_1,\myMat{J}_2$\;

    Compute  source assignment $\myMat{P}^{\star}$ using \eqref{eqn:optimal_assignment_ussn}\;

    Compute the \ac{doa} loss
    $\mathcal{L}^{(1)}_{\mySet{D}^{(b)}}(\myVec{\psi})$
    using \eqref{eqn:doa_loss_ussn}\;

    Update
    $\myVec{\psi}
    \gets
    \myVec{\psi}
    -
    \eta
    \nabla_{\myVec{\psi}}
    \mathcal{L}^{(1)}_{\mySet{D}^{(b)}}(\myVec{\psi})$\;
} 

}

\nonl\textbf{Stage 2: Uncertainty-Aware Fine-Tuning};

\For{$e=1,\ldots,e_2$}
{
Randomly divide $\mySet{D}$ into $B$ batches ${\mySet{D}^{(b)}}_{b=1}^{B}$; 

\For{each batch $\mySet{D}^{(b)}$}
{
    Compute
    $\hat{\myVec{\theta}}$ and $\hat{\myMat{\Sigma}}$
    for all $(\myMat{X},\myVec{\theta})\in\mySet{D}^{(b)}$\;

    Compute  source assignment $\myMat{P}^{\star}$ using \eqref{eqn:optimal_assignment_ussn}\;

    Form the matched angular error $\myVec{e}_{\theta}$ using \eqref{eqn:matched_error_vector}\;

    Permute 
    $\bar{\myMat{\Sigma}}
    =
    \myMat{P}^{\star}
    \hat{\myMat{\Sigma}}
    (\myMat{P}^{\star})^{\mathsf T}$\;

    Form the regularized covariance
    $\bar{\myMat{\Sigma}}_{\epsilon}$
    using \eqref{eqn:regularized_covariance}\;

    Compute  $\ell_{\rm DoA}$ using \eqref{eqn:doa_loss_ussn} and  $\ell_{\rm UQ}$ using \eqref{eqn:uq_loss_ussn}\;

    Compute the batch loss
    $\mathcal{L}^{(2)}_{\mySet{D}^{(b)}}(\myVec{\psi})$
    using \eqref{eqn:stage2_loss_ussn}\;

    Update
    $\myVec{\psi}
    \gets
    \myVec{\psi}
    -
    \eta
    \nabla_{\myVec{\psi}}
    \mathcal{L}^{(2)}_{\mySet{D}^{(b)}}(\myVec{\psi})$\;
} 
}

\KwRet{$\myVec{\psi}$}
\end{algorithm}

\subsection{Discussion} 
\label{ssec:Discssion} 
\noindent
{\bf Properties:} 
The proposed framework addresses the uncertainty-aware \ac{doa} estimation problem formulated in Subsection~\ref{ssec:Problem} by combining data-aided robustness with model-based uncertainty characterization. The surrogate covariance recovery module \(g_{\myVec{\psi}}(\cdot)\) allows the covariance processed by \ac{esprit} to deviate from the empirical covariance \(\hat{\myMat{R}}_X\), thereby enabling the learned representation to compensate for the challenging conditions \ref{itm:Coherent}-\ref{itm:Snapshots}. At the same time, the subsequent \ac{esprit} processing preserves the subspace structure needed for perturbation-based uncertainty extraction. The output covariance \(\hat{\myMat{\Sigma}}\) thus follows from the same surrogate covariance and \ac{esprit} quantities used to estimate \(\hat{\myVec{\theta}}\). The two-stage  procedure further adapts the surrogate covariance to support both accurate angle recovery and reliable covariance prediction, while the Gaussian  approximation in \eqref{eqn:gaussian_plugin_moments} enables the required perturbation statistics to be computed from the recovered covariance at inference time.

\smallskip
\noindent
{\bf Complexity:}  
Standard \ac{esprit}, when applied to an available covariance estimate, is dominated by the \ac{evd} of the \(N\times N\) covariance matrix and the \(M\times M\) rotational-invariance eigendecomposition, yielding complexity on the order of \(\mathcal{O}(N^3+NM^2+M^3)\), in addition to the cost of forming the empirical covariance, \(\mathcal{O}(N^2T)\). SubspaceNet with pointwise \ac{doa} estimation replaces the empirical covariance by the learned surrogate covariance, adding the \ac{dnn} forward-pass cost \(C_{\rm DNN}\), but otherwise retaining the same \ac{esprit} complexity. Algorithm~\ref{alg:ussn_inference} uses the same \ac{dnn} output and reuses the subspace quantities computed by \ac{esprit}; its additional cost comes from the covariance-reconstruction block in \eqref{eq:cov_moments}-\eqref{eqn:ussn_covariance_output}. When implemented using the Gaussian plug-in structure without explicitly storing fourth-order perturbation tensors, this additional deterministic post-processing scales as \(\mathcal{O}(M^2N^3+M^3N^2)\), which reduces to \(\mathcal{O}(M^2N^3)\) as \(M\leq N\). Thus, compared with SubspaceNet point estimation, the proposed method increases only the model-based post-processing cost, without requiring additional \ac{dnn} branches or ensembles.

\smallskip
\noindent
{\bf Potential Extensions:}   
Our characterization is derived for \ac{esprit}, and is therefore naturally suited to arrays admitting a rotational-invariance structure, such as \acp{ula} or arrays that can be partitioned into shift-invariant subarrays. Still, the surrogate covariance recovery  underlying SubspaceNet is not inherently restricted to \acp{ula}. Extending the perturbation-based covariance characterization to other subspace estimators and array geometries, including general planar arrays, sparse arrays, and near-field  settings, is  left  for future research. Another direction combines the proposed model-based covariance recovery with learning-based uncertainty quantification tools. For example, conformal prediction could be used to calibrate the coverage of the covariance-induced confidence regions~\cite{weisman2026conformal}, while Bayesian \acp{dnn}   could be used to account for uncertainty in the learned surrogate covariance~\cite{dahan2025bayesian}. Such combinations may provide complementary uncertainty information while retaining the  interpretable covariance structure induced by our \ac{esprit}-based analysis.

\section{Numerical Study}
\label{sec: Numerical Study}

\input{plotting_latex_code/non_coherent_snr_sweep_figures_pgfplots.tex}

We numerically evaluate the proposed Uncertainty-Aware SubspaceNet (termed here USSN) framework in terms of both \ac{doa} estimation accuracy and uncertainty reliability\footnote{The source code and full hyperparameters can be found at at \url{https://github.com/RazZohar/Uncertainty_SubspaceNet.git}.}. The experiments are designed to examine whether USSN preserves the favorable behavior of classical subspace processing in nominal regimes, while maintaining reliable covariance prediction under the challenging conditions \ref{itm:Coherent}-\ref{itm:Snapshots}. 

\input{plotting_latex_code/non_coherent_eta_sweep_figures_pgfplots.tex}

\subsection{Experimental Setup}
\label{subsec:Numerical Study Setup}

\subsubsection{Signal Generation}
Unless stated otherwise, we consider a \ac{ula} with $N=8$ elements and generate $T=100$ snapshots from $M=2$ sources. The source \acp{doa} are drawn uniformly over the angular sector $[-\frac{\pi}{2},\frac{\pi}{2}]$, and the  \ac{snr} is set to $10$~dB. The source symbols and additive noise are modeled as i.i.d. circularly symmetric complex Gaussian, except in experiments explicitly considering non-Gaussian sources. In coherent-source scenarios, the sources transmit fully correlated signals. Array miscalibration is modeled through a perturbation parameter $\eta$, which controls the deviation of the true array response from the nominal half-wavelength \ac{ula} steering model in \eqref{eqn:SteeringVec}.

\input{plotting_latex_code/non_coherent_M_sweep_figures_pgfplots.tex}

\subsubsection{Benchmarks}
We compare USSN with both classical and learning-based baselines. The classical baselines include standard \ac{esprit} with covariance-based uncertainty extraction, as well as the conditional \ac{crb} when available. The latter is used as a reference for achievable accuracy and uncertainty under the corresponding idealized statistical model.

\input{plotting_latex_code/fixed_anchor_sweep_snr0}
The learning-based baselines include SubspaceNet trained only for \ac{doa} accuracy, denoted SubspaceNet Stage~1, and the proposed USSN obtained after uncertainty-aware fine-tuning. We also compare with data-driven uncertainty baselines, including a TransMUSIC-based architecture~\cite{ji2024transmusic} augmented with an uncertainty-prediction head, and a  data-driven \ac{cnn}  that maps received  samples to \acp{doa}~\cite{zheng2024deepdoa}, which we adapt and train to jointly output \acp{doa}
and their associated uncertainty estimates.  In the ablation study, we further compare with uncertainty quantification based on Monte-Carlo dropout (termed here {\em MC}) as in~\cite{fu2026deep} and conformal prediction  (termed {\em CP}) following \cite{khurjekar2023uncertainty}. All learning-based methods are trained and evaluated using the same datasets within each experimental scenario.

\input{plotting_latex_code/modulation_16qam_snr_sweep_figures_pgfplots.tex}

\subsubsection{Evaluation Metrics}
We evaluate \ac{doa} recovery accuracy using the \ac{rmspe} in \eqref{eqn:doa_loss_ussn}. The uncertainty prediction is evaluated using the uncertainty loss in \eqref{eqn:uq_loss_ussn}, which accounts for the predicted covariance and the observed \ac{doa} estimation error. Lower values of both metrics indicate improved performance.
In addition, we use the normalized \ac{anees} as a calibration diagnostic for the predicted covariance~\cite{dahan2025bayesian}. For a test set of $K$ realizations, let $\myVec{e}_{\theta}^{(k)}$ denote the matched periodic \ac{doa} error vector of the $k$th realization, and let $\bar{\myMat{\Sigma}}_{\epsilon}^{(k)}$ be its corresponding matched  covariance estimate. We compute the \ac{anees} as
\begin{equation}
\operatorname{ANEES} = \frac{1}{KM} \sum_{k=1}^{K} \left( \myVec{e}_{\theta}^{(k)} \right)^{\mathsf T} \left( \bar{\myMat{\Sigma}}_{\epsilon}^{(k)} \right)^{-1} \myVec{e}_{\theta}^{(k)}.
\label{eqn:nanees_metric}
\end{equation}
A well-calibrated covariance estimator is expected to yield $\operatorname{ANEES}\approx 1$. In the figures, we report the logarithm of this quantity, such that the ideal value is $\log(1)=0$. Unless stated otherwise, all metrics are averaged over $K=5000$  test samples.
The dashed reference curve in the three-panel benchmark figures denotes the \ac{crb} in the \ac{doa} accuracy panel, the \ac{crb}-based uncertainty reference in the uncertainty-loss panel, and the ideal value $\log(1)=0$ in the normalized-\ac{anees} panel. 

\input{plotting_latex_code/coherent_snr_sweep_figures_pgfplots.tex}

%
\input{plotting_latex_code/coherent_eta_sweep_figures_pgfplots.tex}
%

\input{plotting_latex_code/ssn_fulltrain_vs_cp_stage1_vs_stage2_uq_figures_pgfplots.tex}

\subsection{Results}
\label{subsec:Numerical Study Results}

\subsubsection{Nominal Non-Coherent Sources}
We begin with the nominal setting of non-coherent Gaussian sources and a calibrated array. This experiment serves as a sanity check, since classical subspace methods and their associated uncertainty analyses are expected to be reliable when their modeling assumptions hold.
The results in \figref{fig:non_coherent_snr_benchmark_combined} show that USSN preserves accurate \ac{doa} estimation across the considered \ac{snr} range while also producing reliable uncertainty estimates. In this regime, classical \ac{esprit} provides a meaningful reference, since the source covariance is full rank and the array model is correctly specified. Nevertheless, the proposed uncertainty-aware training improves the calibration of the covariance output compared with SubspaceNet trained only for pointwise \ac{doa} recovery.  

\subsubsection{Robustness to Array Miscalibration}
We next evaluate robustness to array miscalibration by varying the perturbation parameter $\eta$. This experiment targets  \ref{itm:Miscalibration}, where the true array response deviates from the nominal steering model used by classical \ac{esprit}.
As shown in \figref{fig:non_coherent_eta_benchmark_doa}, the performance of classical \ac{esprit} deteriorates as the array perturbation increases, since the subspace estimator relies on an inaccurate steering model. In contrast, USSN remains substantially more robust, as the learned surrogate covariance is optimized to support the downstream \ac{esprit} operation despite the mismatch. The uncertainty metrics show a similar trend: the proposed method provides covariance estimates that remain informative under miscalibration, whereas classical uncertainty extraction becomes unreliable when the assumptions underlying the perturbation analysis are violated.

\subsubsection{Effect of the Number of Sources}
We also evaluate the effect of increasing the number of sources. This setting is particularly relevant here since USSN estimates a full $M\times M$ error covariance matrix whose off-diagonal entries characterize coupling between different \ac{doa} errors.
The results in \figref{fig:non_coherent_m_benchmark_doa} demonstrate that USSN maintains reliable \ac{doa} estimation and uncertainty characterization as the number of sources increases. As expected, the problem becomes more difficult for larger $M$, since the angular separation between sources tends to decrease and the dimension of the covariance output grows. Still, USSN remains competitive in accuracy while providing more reliable covariance estimates than baselines that do not exploit the \ac{esprit} structure.

\subsubsection{Geometry-Dependent Uncertainty}
To further examine whether the proposed covariance estimate reflects the geometry of the source configuration, we consider a two-source setting at $\ac{snr}=0$~dB. One source is fixed at $\theta_1=40^{\circ}$, while the second source is swept across the angular sector, with finer sampling near the fixed source.
The results in \figref{fig:all_models_fixed_anchor_sweep_snr0_accuracy} and \figref{fig:all_models_fixed_anchor_sweep_snr0_variance} show that the estimation difficulty varies significantly with the angular position of the moving source. The error increases near endfire and when the two sources become closely spaced, reflecting the reduced ability to resolve the sources in these configurations.  USSN captures these changes in its predicted uncertainty. This experiment highlights the value of a sample-dependent covariance estimate: the uncertainty is not determined only by global parameters such as \ac{snr}, but also by the specific source geometry and the coupling between the estimated directions.

\subsubsection{Non-Gaussian Sources}
We next examine a non-Gaussian signaling scenario using $16$-QAM source symbols. The motivation from this study stems from the fact that the covariance perturbation moments used by USSN are computed with a Gaussian plug-in approximation, whereas communication signals are  drawn from discrete constellations.
\figref{fig:modulation_16qam_snr_benchmark_combined} indicates that USSN remains effective under non-Gaussian signaling. Although the plug-in perturbation approximation is derived under a local Gaussian assumption, the learned surrogate covariance and uncertainty-aware training enable the method to provide useful uncertainty estimates for communication-like source signals. This suggests that the proposed framework is not limited to ideal Gaussian source models, and can be applied in more practical array processing and wireless communication settings.

\subsubsection{Coherent Sources and Miscalibration}
We now turn to coherent sources, where the source covariance matrix is rank deficient, violating the condition enabling conventional subspace separation.
The  results, reported in \figref{fig:coherent_snr_benchmark_combined},  demonstrate a key advantage of our framework. Classical \ac{esprit} struggles here, and its  uncertainty extraction is no longer reliable because both the estimator and the perturbation characterization rely on assumptions that are violated. USSN, however, uses the learned surrogate covariance to restore a representation suitable for subspace processing, while the uncertainty-aware training encourages this representation to also support covariance reconstruction. As a result, USSN provides both accurate \ac{doa} estimates and meaningful uncertainty measures across the considered \ac{snr} range.

We also consider the combined effect of source coherence and array miscalibration. This represents a  challenging regime, since both the full-rank source covariance assumption and the nominal array-response assumption are violated simultaneously.
The results reported in Fig.~\ref{fig:coherent_eta_benchmark_combined} show that the degradation of classical \ac{esprit} is further amplified when coherence and miscalibration occur together.  USSN remains robust across the perturbation range, demonstrating the benefit of learning a surrogate covariance that is optimized for the downstream \ac{esprit} estimator. The uncertainty estimates produced by USSN also remain informative in this combined-mismatch setting, supporting the claim that our method can provide uncertainty-aware \ac{doa} recovery beyond the standard operating assumptions of  subspace methods.

\subsubsection{Uncertainty Ablation Study}
\label{sssec: Ablation Study}
We conclude with an ablation study comparing our pertubation-based formulation with generic deep-learning uncertainty quantification techniques. We compare USSN with Monte-Carlo dropout, implementing a Bayesian \ac{dnn}, and with conformal prediction following the calibration-based localization approach. The comparison is carried out in the non-coherent setting with $M=2$ sources under an \ac{snr} sweep.
The results in \figref{fig:fulltrain-cp-accuracy-coverage} show that USSN provides a favorable accuracy--reliability tradeoff. Conformal prediction provides calibrated marginal intervals, but these intervals are not sample-dependent full covariance matrices. Similarly, Monte-Carlo dropout reflects variability induced by the learned predictor, but does not directly characterize the perturbation sensitivity of the subspace estimator. In contrast, USSN produces a covariance matrix tied to the surrogate covariance and the retained \ac{esprit} computation, which enables it to capture the dependence of uncertainty on the observed sample and source configuration.

The additional diagnostics in \figref{fig:fulltrain-cp-qhat-nll} further illustrate this distinction. The conformal half-widths reflect calibration of marginal residuals, while the proposed method produces a native full-covariance estimate. Consequently, metrics such as Gaussian negative log-likelihood and normalized \ac{anees} are naturally aligned with the USSN output, whereas they can only be applied to conformal intervals through a surrogate diagonal-Gaussian interpretation. These results support our design principle, that uncertainty is most effective when it is extracted through the preserved model-based processing structure.

\section{Conclusion}
\label{sec: Conclusions}
We introduced an uncertainty extraction method for  SubspaceNet-based \ac{doa} estimation, which leverages its preserved subspace structure to provide  uncertainty quantification. By adapting asymptotic ESPRIT error analysis and identifying suitable approximations based on features provided by SubspaceNet, our method yields faithful uncertainty estimates alongside accurate \ac{doa} recovery. Numerical results confirm that the proposed approach enables accurate estimation of both the \ac{doa} and the error level in various challenging settings. 

\appendices


\section{Proof of Theorem~\ref{thm:full_doa_covariance_overlapping_esprit}}
\label{app:proof_full_doa_covariance}

We prove Theorem~\ref{thm:full_doa_covariance_overlapping_esprit} by propagating the
covariance perturbation through three mappings: the covariance-to-subspace mapping, the
subspace-to-ESPRIT-eigenvalue mapping, and the ESPRIT-eigenvalue-to-\ac{doa} mapping.  
Our methodology is based on characterizing the perturbations between quantities computed by \ac{esprit} applied to the {\em true} input covariance $\myMat{R}_X$ (e.g., $\myMat{F}, \myMat{E}_x, \myMat{E}_y$, $\myVec{\lambda}$, etc.), and the corresponding quantities denoted by \ac{esprit} applied to the estimated covariance $\hat{\myMat{R}}_X$, for which all quantities are denoted with $\hat{(\cdot)}$ (e.g., $\hat{\myMat{F}}, \hat{\myMat{E}}_x, \hat{\myMat{E}}_y$,  $\hat{\myVec{\lambda}}$, etc.). As in \eqref{eqn:covariance_perturbation}, we write the differences between these quantities with $\Delta$'s (e.g., $\Delta{\myMat{F}}, \Delta{\myMat{E}}_x, \Delta{\myMat{E}}_y$, $\Delta\myVec{\lambda}$, etc.).

Following the above rationale, let $ \hat{\myMat{E}}_{\rm S}
    = \myMat{E}_{\rm S}     +     \Delta \myMat{E}_{\rm S}$ be the 
signal subspace recovered from $\hat{\myMat{R}}_X$. 
Under the eigenvalue separation assumption, standard first-order subspace perturbation theory \cite[Ch. 4-5]{stewart1990matrix}
gives a linear relation between the covariance perturbation and the signal-subspace perturbation:
\begin{align}
    \operatorname{vec}
    \left(
    \Delta \myMat{E}_{s}
    \right)
    =
    \myMat{T}_{E}
    \operatorname{vec}
    \left(
    \Delta \myMat{R}
    \right)
    +
    \mathcal{O}
    \left(
    \left\|
    \Delta \myMat{R}
    \right\|^2
    \right),
    \label{eqn:subspace_perturbation_linear_map}
\end{align}
where \(\myMat{T}_{E}\) is the first-order subspace perturbation operator \cite[Part V]{stewart1990matrix}.
Now, recall that \ac{esprit} aims at computing $\myMat{E}_{x}
    =
    \myMat{J}_{1}\myMat{E}_{\rm S}$ and $\myMat{E}_{y}
    =    \myMat{J}_{2}\myMat{E}_{\rm S}$, for which  
\begin{align}
    \myMat{E}_{y}
    =
    \myMat{E}_{x}\myMat{F}.
    \label{eqn:exact_shift_invariance}
\end{align}
Accordingly, when applying \ac{esprit} to $\hat{\myMat{R}}_X$, we have that \eqref{eqn:exact_shift_invariance} gives
   $\hat{\myMat{E}}_y = \myMat{E}_{y}
    +
    \Delta \myMat{E}_{y}
    = 
    \hat{\myMat{E_x}}\hat{\myMat{F}}
    =
    \left(
    \myMat{E}_{x}
    +
    \Delta \myMat{E}_{x}
    \right)
    \left(
    \myMat{F}
    +
    \Delta \myMat{F}
    \right)$.
Analyzing first-order error while keeping  second-order terms with $\mathcal{O}$ notations, results in 
\begin{align}
    \Delta \myMat{E}_{y}
    =
    \Delta \myMat{E}_{x}\myMat{F}
    +
    \myMat{E}_{x}\Delta \myMat{F} 
    +
    \mathcal{O}
    \left(
    \left\|
    \Delta \myMat{R}
    \right\|^2
    \right),
    \label{eqn:first_order_shift_invariance}
\end{align}
Since $ \Delta \myMat{E}_{x}  =  \myMat{J}_{1}\Delta \myMat{E}_{\rm S}$, and $\Delta \myMat{E}_{y}   =
    \myMat{J}_{2}\Delta \myMat{E}_{\rm S}$,
 \eqref{eqn:first_order_shift_invariance} implies that
\begin{align}
    \myMat{J}_{2}\Delta \myMat{E}_{\rm S}
    =
    \myMat{J}_{1}\Delta \myMat{E}_{\rm S}\myMat{F}
    +
    \myMat{E}_{x}\Delta \myMat{F} +
    \mathcal{O}
    \left(
    \left\|
    \Delta \myMat{R}
    \right\|^2
    \right).
\end{align}
Left-multiplying by \(\myMat{E}_{x}^{\dagger}\) yields 
\begin{align}
    \Delta \myMat{F}
    =
    \myMat{E}_{x}^{\dagger}
    \left(
    \myMat{J}_{2}\Delta \myMat{E}_{s}
    -
    \myMat{J}_{1}\Delta \myMat{E}_{s}\myMat{F}
    \right)+
    \mathcal{O}
    \left(
    \left\|
    \Delta \myMat{R}
    \right\|^2
    \right).
    \label{eqn:psi_perturbation_first_form}
\end{align}
For an eigenvalue \(\lambda_i\) of order one,
since $\hat{\myVec{v}}_i$ is right eigenvector of $\hat{\myMat{F}}$,  we obtain the effect of eigenvalue perturbation as
\begin{flalign}
\hat{\myMat{F}}\hat{\myVec{v}}_i &= (\myMat{F}+ \Delta{\myMat{F}})(\myVec{v}_i+\Delta{\myVec{v}}_i) 
\notag \\&=
(\lambda_i + \Delta{\lambda}_i)(\myVec{v}_i + \Delta{\myVec{v}}_i)
=  \hat{\lambda}_i\hat{\myVec{v}}_i.
\label{eqn:eigenvalue_perturbation_appendix_full}
\end{flalign}
Then left multiplication by $\myVec{q}_i^{\mathsf{H}}$ and requiring normalization $\myVec{q}_i^{\mathsf{H}}\myVec{v}_i = 1$ yields the first-order eigenvalue perturbation 
\begin{align}
    \Delta \lambda_i
    =
    \myVec{q}_{i}^{H}
    \Delta \myMat{F}
    \myVec{v}_{i} +
    \mathcal{O}
    \left(
    \left\|
    \Delta \myMat{R}
    \right\|^2
    \right).
\label{eqn:eigenvalue_perturbation_appendix}
\end{align}
Substituting \eqref{eqn:psi_perturbation_first_form} into
\eqref{eqn:eigenvalue_perturbation_appendix} gives
\begin{align*}
    \Delta \lambda_i
    &=
    \myVec{q}_{i}^{\mathsf{H}}
    \myMat{E}_{x}^{\dagger}
    \left(
    \myMat{J}_{2}\Delta \myMat{E}_{\rm S}
    -
    \myMat{J}_{1}\Delta \myMat{E}_{\rm S}\myMat{F}
    \right)
    \myVec{v}_{i} +
    \mathcal{O}
    \left(
    \left\|
    \Delta \myMat{R}
    \right\|^2
    \right).
\end{align*}
Since $ \myMat{F}\myVec{v}_{i}   = \lambda_i\myVec{v}_{i}$,
we obtain
\begin{align}
    \Delta \lambda_i
    &=
    \myVec{q}_{i}^{\mathsf{H}}
    \myMat{E}_{x}^{\dagger}
    \left(
    \myMat{J}_{2}
    -
    \lambda_i\myMat{J}_{1}
    \right)   \Delta \myMat{E}_{\rm S}
    \myVec{v}_{i}+
    \mathcal{O}
    \left(
    \left\|
    \Delta \myMat{R}
    \right\|^2
    \right).
    \label{eqn:delta_lambda_overlapping_appendix}
\end{align}
Equation~\eqref{eqn:delta_lambda_overlapping_appendix} shows that each ESPRIT eigenvalue
perturbation is a linear function of the signal-subspace perturbation $ \Delta \myMat{E}_{\rm S}$. Combining
\eqref{eqn:subspace_perturbation_linear_map} and
\eqref{eqn:delta_lambda_overlapping_appendix} for all \(i=1,\ldots,M\), we may write
\begin{align}
    \Delta \myVec{\lambda}
    =
    \myMat{T}_{\lambda}
    \operatorname{vec}
    \left(
    \Delta \myMat{R}
    \right)
    +
    \mathcal{O}
    \left(
    \left\|
    \Delta \myMat{R}
    \right\|^2
    \right),
    \label{eqn:eigenvalue_perturbation_linear_map}
\end{align}
where \(\myMat{T}_{\lambda}\) is the resulting first-order covariance-to-eigenvalue
perturbation operator.

Next, we propagate the eigenvalue perturbation to the \ac{doa} perturbation. The ideal ESPRIT eigenvalue associated with the $i$th source is $\lambda_i =
\exp\left(
-j\pi\sin\theta_i
\right)$. 
Accordingly, for applying \ac{esprit} to $\hat{\myMat{R}}_X$, we obtain the eigenvalues $\{\hat{\lambda}_i\}$ such that 
\begin{equation}
    \Delta \lambda_i = \hat{\lambda}_i -  \exp
    \left(
    -j\pi\sin\theta_i
    \right), \qquad i \in \{1,\ldots,M\},
    \label{eqn:Eigen_pertubation}
\end{equation}
whose stacking is $\Delta\myVec{ \lambda} = [\Delta \lambda_1,\ldots \Delta \lambda_M]^{\mathsf{T}}$. 
Recall that the first-order Taylor series expansion of the complex phase function $\angle(\cdot)$ around $z_0\in \mathbb{C}$ is given by $\angle(z) \approx \angle(z_0) +\Im\left\{\frac{z-z_0}{z_0} \right\}$, where $\Im\{\cdot\}$ is the imaginary part of a complex number~\cite{carrier2005functions}. Accordingly, 
for a small eigenvalue perturbation, the first-order Taylor series expansion of $\angle(\cdot)$ around $\lambda_i$ yields
\begin{align}
    \Delta \angle \lambda_i &\triangleq \angle\left(\lambda_i + \Delta \lambda_i\right) - \angle\lambda_i \notag \\
    &= \angle\frac{\hat{\lambda_i}} {\lambda_i}= \angle{\frac{\lambda_i + \Delta\lambda_i}{\lambda_i}} = \angle(1 + \frac{\Delta \lambda_i}{\lambda_i})
    \notag \\&=
    \Im
    \left\{
    \frac{\Delta \lambda_i}{\exp
    \left(
    -j\pi\sin\theta_i
    \right)}
    \right\}
    +
    \mathcal{O}
    \left(
    \left\|
    \Delta \myMat{R}
    \right\|^2
    \right).
\end{align}
since $\theta_i = \sin^{-1}\!\big(- \tfrac{\angle\,{\lambda}_i}{\pi} \big)$, we have first-order Taylor expansion
\begin{flalign}
    \Delta \theta_i &= \frac{d\theta_i}{d(\angle\lambda_i)}\Delta\angle{\lambda_i} +
    \mathcal{O}
    \left(
    \left\|
    \Delta \myMat{R}
    \right\|^2
    \right)
    \notag \\ 
    &= \frac{d\theta_i}{d(\sin\theta_i)}\frac{d(\sin\theta_i)}{d(\angle\lambda_i)}\Delta\angle{\lambda_i}+
    \mathcal{O}
    \left(
    \left\|
    \Delta \myMat{R}
    \right\|^2
    \right)
    \notag \\ &=
    \frac{-1}{\pi\cos\theta_i}
    \Im
    \left\{
    \frac{\Delta \lambda_i}{\exp
    \left(
    -j\pi\sin\theta_i
    \right)}
    \right\}
    \!+\!
    \mathcal{O}
    \left(
    \left\|
    \Delta \myMat{R}
    \right\|^2
    \right).
    \label{eqn:theta_perturbation_appendix}
\end{flalign}

Using the matrix $ \myMat{\Lambda}_{\myVec{\theta}}$ in \eqref{eqn:lambdaDef}, we define the vector $ \Delta \myVec{z} \triangleq 
    \myMat{\Lambda}_{\myVec{\theta}}^{-1}
    \Delta \myVec{\lambda}$ (whose $i$th entry is $\frac{\Delta \lambda_i}{\exp
    \left(
    -j\pi\sin\theta_i
    \right)}$).
Then, since there are $M$ sources, we get a unified equation using the matrix $\myMat{D}_{\myVec{\theta}}$ defined in \eqref{eqn:DthetaDef}, i.e., \eqref{eqn:theta_perturbation_appendix} can be written compactly as
\begin{align}
    \Delta \myVec{\theta}
    =
    -
    \myMat{D}_{\theta}
    \Im
    \left\{
    \Delta \myVec{z}
    \right\}
    +
    \mathcal{O}
    \left(
    \left\|
    \Delta \myMat{R}
    \right\|^2
    \right).
    \label{eqn:theta_vector_perturbation_appendix}
\end{align}
Accounting for the higher-order error term originating from \eqref{eqn:eigenvalue_perturbation_linear_map}, the vector expansion yields:
\begin{align}
    \Delta \myVec{z} 
    &= \myMat{\Lambda}_{\myVec{\theta}}^{-1} \myMat{T}_{\lambda} \operatorname{vec} \left( \Delta \myMat{R} \right) 
   +
    \mathcal{O}
    \left(
    \left\|
    \Delta \myMat{R}
    \right\|^2
    \right).
    \label{eqn:DeltaZ}
\end{align}

Next,  define the eigenvalue perturbation covariance matrices 
\begin{equation}
    \myMat{K}_{\lambda}
    =
    \mathbb{E}
    \left[
    \Delta \myVec{\lambda}
    \Delta \myVec{\lambda}^{\mathsf{H}}
    \right], \quad \widetilde{\myMat{K}}_{\lambda}
    =
    \mathbb{E}
    \left[
    \Delta \myVec{\lambda}
    \Delta \myVec{\lambda}^{\mathsf{T}}
    \right].
    \label{eqn:Kmats}
\end{equation}
It is noted that when all the eigenvalues of $\myMat{F}$ are of order one, then the covariance matrices in \eqref{eqn:Kmats} are given by \eqref{eq:Klambda_pairwise_full}, see \cite[Sec. III.A]{yuen2002asymptotic}.
The covariance and pseudo-covariance of the principal first-order vector term \(\Delta \myVec{z}\) are
\begin{subequations}
    \label{eqn:z_covariance_app}
\begin{align}
    \myMat{K}_{z}
    &=
    \mathbb{E}
    \left[
    \Delta \myVec{z}
    \Delta \myVec{z}^{\mathsf{H}}
    \right]
    =
    \myMat{\Lambda}_{\myVec{\theta}}^{-1}
    \myMat{K}_{\lambda}
    \myMat{\Lambda}_{\myVec{\theta}}^{-\mathsf{H}},
    \label{eqn:z_covariance_appendix}
    \\
    \widetilde{\myMat{K}}_{z}
    &=
    \mathbb{E}
    \left[
    \Delta \myVec{z}
    \Delta \myVec{z}^{\mathsf{T}}
    \right]
    =
    \myMat{\Lambda}_{\myVec{\theta}}^{-1}
    \widetilde{\myMat{K}}_{\lambda}
    \myMat{\Lambda}_{\myVec{\theta}}^{-\mathsf{T}}.
    \label{eqn:z_pseudocovariance_appendix}
\end{align}
\end{subequations}
For a complex  \(\Delta \myVec{z}\), the covariance of its imaginary part is
\begin{align}
    \mathbb{E}
    \left[    \Im\{\Delta \myVec{z}\}
    \Im\{\Delta \myVec{z}\}^{\mathsf{T}}
    \right]
    =
    \frac{1}{2}
    \Re
    \left\{
    \myMat{K}_{z}
    -
    \widetilde{\myMat{K}}_{z}
    \right\}.
    \label{eqn:imaginary_covariance_identity}
\end{align}
Expanding the outer product of the total \ac{doa} error vector from \eqref{eqn:theta_vector_perturbation_appendix}, the cross-multiplication of the deterministic bounds scales the approximation error quadratically 
\begin{align}
    \Delta \myVec{\theta} \Delta \myVec{\theta}^{\mathsf{T}} 
    &= \left( -\myMat{D}_{\myVec{\theta}} \Im \left\{ \Delta \myVec{z} \right\} +
    \mathcal{O}
    \left(
    \left\|
    \Delta \myMat{R}
    \right\|^2
    \right) \right) \nonumber \\
    &\quad  \left( -\myMat{D}_{\myVec{\theta}} \Im \left\{ \Delta \myVec{z} \right\} + \mathcal{O}
    \left(
    \left\|
    \Delta \myMat{R}
    \right\|^2
    \right)\right)^{\mathsf{T}} \nonumber \\
    &= \myMat{D}_{\myVec{\theta}} \Im\left\{\Delta \myVec{z}\right\} \Im\left\{\Delta \myVec{z}\right\}^{T} \myMat{D}_{\myVec{\theta}}  + \mathcal{O}\left(\left\|\Delta \myMat{R}\right\|^3\right).
\end{align}
%
Applying stochastic expectation to both sides yields
\begin{align}
   \myMat{\Sigma} &=\mathbb{E}
    \left[
    \Delta \myVec{\theta}
    \Delta \myVec{\theta}^{\mathsf{T}} 
    \right] \nonumber \\
    &=
    \myMat{D}_{\myVec{\theta}}
    \mathbb{E}
    \left[
    \Im\{\Delta \myVec{z}\}
    \Im\{\Delta \myVec{z}\}^{\mathsf{T}} 
    \right]
    \myMat{D}_{\myVec{\theta}}
    +
    \mathcal{O}
    \left(
    \mathbb{E}[\left\|
    \Delta \myMat{R}
    \right\|^{3}]
    \right) \notag\\
    &\stackrel{(a)}{=}
    \frac{1}{2}
    \myMat{D}_{\myVec{\theta}}
    \Re
    \left\{
    \myMat{K}_{z}
    -
    \widetilde{\myMat{K}}_{z}
    \right\}
    \myMat{D}_{\myVec{\theta}}
    +
    \mathcal{O}
    \left(
    \mathbb{E}[\left\|
    \Delta \myMat{R}
    \right\|^{3}]
    \right) \notag \\
    &\stackrel{(b)}{=}
    \frac{1}{2}
    \myMat{D}_{\myVec{\theta}}
    \Re
    \Big\{ \myMat{C}_{{\myVec{\theta}}}
    \Big\}
    \myMat{D}_{\myVec{\theta}}
    +
    \mathcal{O}
    \left(
    \mathbb{E}[\left\|
    \Delta \myMat{R}
    \right\|^{3}]
    \right).
\end{align}
where $(a)$ follows from \eqref{eqn:imaginary_covariance_identity}, and $(b)$ stems from \eqref{eqn:z_covariance_app} combined with the definition of $\myMat{C}_{{\myVec{\theta}}}$ in \eqref{eqn:propagated_eigenvalue_covariance}.  
%
This coincides with \eqref{eqn:full_doa_covariance_overlapping_esprit}, which completes the proof.
\hfill\(\blacksquare\)

\section{Proof of Corollary~\ref{cor:nonoverlap_variance_special_case}}
\label{app:nonoverlap_variance_special_case}
 The corollary follows from Theorem~\ref{thm:full_doa_covariance_overlapping_esprit} by Taking the \(i\)-th diagonal entry gives
\begin{align}
    \left[
    \myMat{\Sigma}
    \right]_{i,i}
    \!=\!
    \frac{1}{2\pi^{2}\cos^{2}\theta_i}
    \left(
    \left[
    \myMat{K}_{\lambda}
    \right]_{i,i}
   \! \!-
     \Re
    \left\{
    \left[
    \widetilde{\myMat{K}}_{\lambda}
    \right]_{i,i}
    \left(\lambda_i^{*}\right)^2
    \right\}
    \right),
    \label{eqn:variance_from_full_covariance}
\end{align}
where we used \(|\lambda_i|=1\), and therefore
\(\lambda_i^{-2}=(\lambda_i^{*})^2\).
The eigenvalue perturbation covariance
terms are given by \eqref{eq:Klambda_pairwise_full}.  
Substituting \eqref{eq:Klambda_pairwise_full} into
\eqref{eqn:variance_from_full_covariance} yields  
\begin{align}
\mathrm{var}(\hat{\theta}_i)
&=\! \frac{1}{2\pi^{2}\cos^{2}\theta_i}
 \Bigg(
 \myVec{q}_i^{\mathsf{H}}
 \myMat{E}_x^{\dagger}
 \myMat{W}_i
 \SigH_i
 \myMat{W}_i^{\mathsf{H}}
 \left(
 \myMat{E}_x^{\dagger}
 \right)^{\mathsf{H}}
 \myVec{q}_i \notag \\
 & \qquad -\! \Re\!
 \left\{
 \myVec{q}_i^{\mathsf{H}}
 \myMat{E}_x^{\dagger}
 \myMat{U}_i
 \SigT_i
 \myMat{U}_i^{\mathsf{T}}
 \left(
 \myMat{E}_x^{\dagger}
 \right)^{\mathsf{T}}
 \myVec{q}_i
 \left(\lambda_i^{*}\right)^2
 \right\}
 \Bigg) \notag \\
 &+ \!\mathcal{O} \left(\mathbb{E}\left[\|\Delta \myMat{R}\|^2\right]\right).
\label{eq:DOA_var_half}
\end{align} 
Hence, the classical variance-only non-overlapping ESPRIT uncertainty expression is
obtained as a special case of the proposed full covariance characterization by retaining
only the diagonal entry \(\left[\myMat{\Sigma}\right]_{i,i}\).
\hfill\(\blacksquare\)

\bibliographystyle{IEEEtran}
\bibliography{IEEEabrv,mybib}

\end{document}

%% file: plotting_latex_code/benchmark_plot_common.tex
\usepackage{cuted}

\usepgfplotslibrary{groupplots}

\definecolor{cESPRIT}{RGB}{31,119,180}
\definecolor{cCCRB}{RGB}{0,0,0}
\definecolor{cSSNStageOne}{RGB}{214,39,40}
\definecolor{cSSNStageTwo}{RGB}{44,160,44}
\definecolor{cTransMUSIC}{RGB}{255,127,14}
\definecolor{cDNNComplex}{RGB}{148,103,189}
\definecolor{cDataDriven}{RGB}{31,119,180}
\definecolor{cMCDropout}{RGB}{214,39,40}
\definecolor{cCP}{RGB}{255,127,14}
\definecolor{cUQSSNStageTwo}{RGB}{44,160,44}

\pgfplotsset{
    benchmarkSweepAxis/.style={
        width=0.94\linewidth,
        height=0.405\linewidth,
        grid=both,
        major grid style={gray!35},
        minor grid style={gray!15},
        tick label style={font=\scriptsize},
        label style={font=\small},
        title style={font=\small, yshift=1.2ex},
        legend style={
            font=\fontsize{4.4}{5.2}\selectfont,
            at={(0.5,1.12)},
            anchor=south,
            legend columns=3,
            /tikz/every even column/.append style={column sep=0.35ex},
            draw=black,
            fill=white,
            inner xsep=1.2pt,
            inner ysep=1.0pt,
            cells={anchor=west},
            row sep=0.1ex,
            text=black,
            every node/.append style={text=black}
        },
        mark size=1.45pt,
        line width=0.72pt,
        unbounded coords=jump,
    }
}

\pgfplotsset{
    benchmarkSubfigureAxis/.style={
        benchmarkSweepAxis,
        width=0.98\linewidth,
        height=0.72\linewidth,
        tick label style={font=\scriptsize},
        label style={font=\footnotesize},
        title style={font=\footnotesize\bfseries},
        mark size=1.35pt,
        line width=0.68pt,
    }
}

\pgfplotsset{
    benchmarkThreePanelAxis/.style={
        benchmarkSweepAxis,
        width=0.310\textwidth,
        height=0.250\textwidth,
        tick label style={font=\scriptsize},
        label style={font=\scriptsize},
        title style={font=\small, yshift=0.15ex},
        legend style={
            font=\fontsize{5.1}{6.0}\selectfont,
            legend columns=-1,
            /tikz/every even column/.append style={column sep=0.35ex},
            draw=black,
            fill=white,
            inner xsep=1.2pt,
            inner ysep=0.8pt,
            cells={anchor=west},
            row sep=0ex,
            text=black,
            every node/.append style={text=black}
        },
        mark size=1.35pt,
        line width=0.68pt,
    }
}

%% file: plotting_latex_code/angle_sweep_setup_snr0.tex
\pgfplotsset{compat=1.18}
\providecolor{cSSNStageTwo}{HTML}{D55E00}
\providecolor{cTransMUSIC}{HTML}{009E73}
\providecolor{cDNNComplex}{HTML}{8B61A8}
\providecolor{cESPRIT}{HTML}{0072B2}
\pgfplotsset{angleSweepZeroAxis/.style={width=\linewidth,height=5cm,grid=major,unbounded coords=jump,tick label style={font=\tiny},label style={font=\scriptsize},scaled ticks=false}}

%% file: plotting_latex_code/non_coherent_snr_sweep_figures_pgfplots.tex

\begin{figure*}
\centering
\scriptsize
{\fontsize{5.2}{6.0}\selectfont
\makebox[\textwidth][c]{%
\textcolor{cCCRB}{\rule{0.9em}{0.7pt}}\,CCRB / ideal reference\hspace{0.70em}%
\textcolor{cESPRIT}{\rule{0.9em}{0.9pt}}\,ESPRIT\hspace{0.70em}%
\textcolor{cSSNStageOne}{\rule{0.9em}{0.9pt}}\,SubspaceNet Stage~1\hspace{0.70em}%
\textcolor{cSSNStageTwo}{\rule{0.9em}{0.9pt}}\,USSN\hspace{0.70em}%
\textcolor{cTransMUSIC}{\rule{0.9em}{0.9pt}}\,TransMUSIC\hspace{0.70em}%
\textcolor{cDNNComplex}{\rule{0.9em}{0.9pt}}\,Data-driven complex%
}
\par}

\vspace{0.25em}
\begin{subfigure}[t]{0.32\textwidth}
\centering
\begin{tikzpicture}
\begin{axis}[benchmarkSubfigureAxis,
xlabel={SNR [dB]},
    xmin=-11.2,
    xmax=21.2,
    xtick={-10,-3,0,3,10,20},
    ylabel={DoA RMSPE [deg]},
    ymode=log,]
\addplot+[forget plot, mark=none, dashed, very thick, color=cCCRB]
    table[x=SNRdB, y=CCRBSigmaDeg, col sep=comma] {numerical_results/non_coherent_snr_sweep_benchmark_pgfplots.csv};

\addplot+[forget plot, mark=*, color=cESPRIT]
    table[x=SNRdB, y=ESPRITDOADeg, col sep=comma] {numerical_results/non_coherent_snr_sweep_benchmark_pgfplots.csv};

\addplot+[forget plot, mark=square*, color=cSSNStageOne]
    table[x=SNRdB, y=SSNStage1DOADeg, col sep=comma] {numerical_results/non_coherent_snr_sweep_benchmark_pgfplots.csv};

\addplot+[forget plot, mark=diamond*, color=cSSNStageTwo]
    table[x=SNRdB, y=SSNStage2DOADeg, col sep=comma] {numerical_results/non_coherent_snr_sweep_benchmark_pgfplots.csv};

\addplot+[forget plot, mark=triangle*, color=cTransMUSIC]
    table[x=SNRdB, y=TransMUSICStage2DOADeg, col sep=comma] {numerical_results/non_coherent_snr_sweep_benchmark_pgfplots.csv};

\addplot+[forget plot, mark=o, color=cDNNComplex]
    table[x=SNRdB, y=DNNComplexStage2DOADeg, col sep=comma] {numerical_results/non_coherent_snr_sweep_benchmark_pgfplots.csv};
\end{axis}
\end{tikzpicture}
\caption{DoA RMSPE.}
\label{fig:non_coherent_snr_benchmark_doa}
\end{subfigure}\hfill
\begin{subfigure}[t]{0.32\textwidth}
\centering
\begin{tikzpicture}
\begin{axis}[benchmarkSubfigureAxis,
xlabel={SNR [dB]},
    xmin=-11.2,
    xmax=21.2,
    xtick={-10,-3,0,3,10,20},
    ylabel={Uncertainty loss},
    ymode=log,]
\addplot+[forget plot, mark=none, dashed, very thick, color=cCCRB]
    table[x=SNRdB, y=CCRBRefUELoss, col sep=comma] {numerical_results/non_coherent_snr_sweep_benchmark_pgfplots.csv};

\addplot+[forget plot, mark=*, color=cESPRIT]
    table[x=SNRdB, y=ESPRITNetUELoss, col sep=comma] {numerical_results/non_coherent_snr_sweep_benchmark_pgfplots.csv};

\addplot+[forget plot, mark=square*, color=cSSNStageOne]
    table[x=SNRdB, y=SSNStage1NetUELoss, col sep=comma] {numerical_results/non_coherent_snr_sweep_benchmark_pgfplots.csv};

\addplot+[forget plot, mark=diamond*, color=cSSNStageTwo]
    table[x=SNRdB, y=SSNStage2NetUELoss, col sep=comma] {numerical_results/non_coherent_snr_sweep_benchmark_pgfplots.csv};

\addplot+[forget plot, mark=triangle*, color=cTransMUSIC]
    table[x=SNRdB, y=TransMUSICStage2NetUELoss, col sep=comma] {numerical_results/non_coherent_snr_sweep_benchmark_pgfplots.csv};

\addplot+[forget plot, mark=o, color=cDNNComplex]
    table[x=SNRdB, y=DNNComplexStage2NetUELoss, col sep=comma] {numerical_results/non_coherent_snr_sweep_benchmark_pgfplots.csv};
\end{axis}
\end{tikzpicture}
\caption{Uncertainty loss.}
\label{fig:non_coherent_snr_benchmark_uncertainty}
\end{subfigure}\hfill
\begin{subfigure}[t]{0.32\textwidth}
\centering
\begin{tikzpicture}
\begin{axis}[benchmarkSubfigureAxis,
xlabel={SNR [dB]},
    xmin=-11.2,
    xmax=21.2,
    xtick={-10,-3,0,3,10,20},
    ylabel={Log norm. ANEES},
    ymajorgrids=true,]
\addplot+[forget plot, mark=none, dashed, very thick, color=cCCRB] coordinates {(-11.2,0) (21.2,0)};

\addplot+[forget plot, mark=*, color=cESPRIT]
    table[x=SNRdB, y=ESPRITLogANEESNormalized, col sep=comma] {numerical_results/non_coherent_snr_sweep_benchmark_pgfplots.csv};

\addplot+[forget plot, mark=square*, color=cSSNStageOne]
    table[x=SNRdB, y=SSNStage1LogANEESNormalized, col sep=comma] {numerical_results/non_coherent_snr_sweep_benchmark_pgfplots.csv};

\addplot+[forget plot, mark=diamond*, color=cSSNStageTwo]
    table[x=SNRdB, y=SSNStage2LogANEESNormalized, col sep=comma] {numerical_results/non_coherent_snr_sweep_benchmark_pgfplots.csv};

\addplot+[forget plot, mark=triangle*, color=cTransMUSIC]
    table[x=SNRdB, y=TransMUSICStage2LogANEESNormalized, col sep=comma] {numerical_results/non_coherent_snr_sweep_benchmark_pgfplots.csv};

\addplot+[forget plot, mark=o, color=cDNNComplex]
    table[x=SNRdB, y=DNNComplexStage2LogANEESNormalized, col sep=comma] {numerical_results/non_coherent_snr_sweep_benchmark_pgfplots.csv};
\end{axis}
\end{tikzpicture}
\caption{Log normalized ANEES.}
\end{subfigure}

\caption{Non-coherent sources, performance measures versus \ac{snr}}
\label{fig:non_coherent_snr_benchmark_combined}
\label{fig:non_coherent_snr_benchmark_log_anees}

\end{figure*}
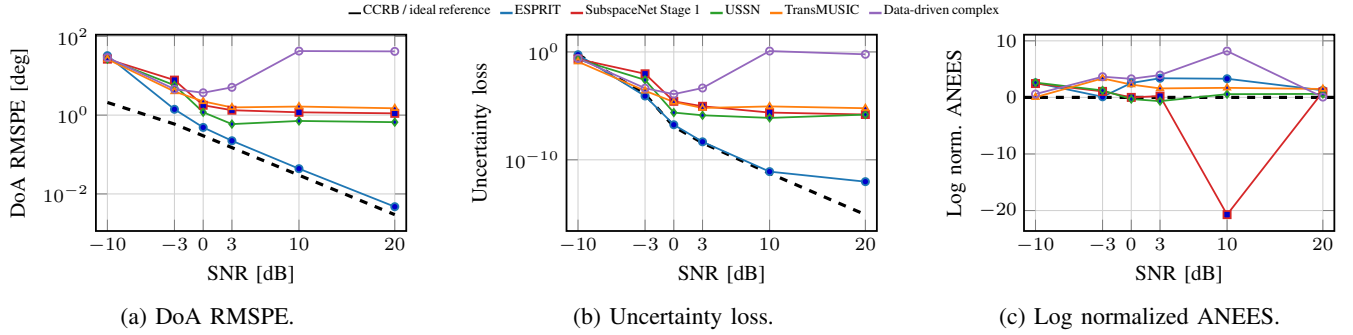

%% file: plotting_latex_code/non_coherent_eta_sweep_figures_pgfplots.tex

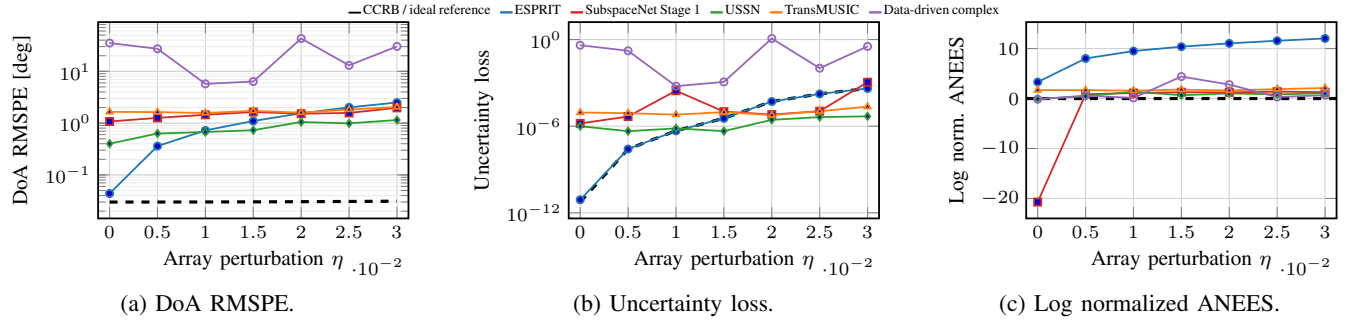
\begin{figure*}[!t]
\centering
\scriptsize
{\fontsize{5.2}{6.0}\selectfont
\makebox[\textwidth][c]{%
\textcolor{cCCRB}{\rule{0.9em}{0.7pt}}\,CCRB / ideal reference\hspace{0.70em}%
\textcolor{cESPRIT}{\rule{0.9em}{0.9pt}}\,ESPRIT\hspace{0.70em}%
\textcolor{cSSNStageOne}{\rule{0.9em}{0.9pt}}\,SubspaceNet Stage~1\hspace{0.70em}%
\textcolor{cSSNStageTwo}{\rule{0.9em}{0.9pt}}\,USSN\hspace{0.70em}%
\textcolor{cTransMUSIC}{\rule{0.9em}{0.9pt}}\,TransMUSIC\hspace{0.70em}%
\textcolor{cDNNComplex}{\rule{0.9em}{0.9pt}}\,Data-driven complex%
}
\par}

\vspace{0.25em}
\begin{subfigure}[t]{0.32\textwidth}
\centering
\begin{tikzpicture}
\begin{axis}[benchmarkSubfigureAxis,
xlabel={Array perturbation $\eta$},
    xmin=-0.0012,
    xmax=0.0312,
    xtick={0,0.005,0.01,0.015,0.02,0.025,0.03},
    ylabel={DoA RMSPE [deg]},
    ymode=log,]
\addplot+[forget plot, mark=none, dashed, very thick, color=cCCRB]
    table[x=Eta, y=CCRBSigmaDeg, col sep=comma] {numerical_results/non_coherent_eta_sweep_benchmark_pgfplots.csv};

\addplot+[forget plot, mark=*, color=cESPRIT]
    table[x=Eta, y=ESPRITDOADeg, col sep=comma] {numerical_results/non_coherent_eta_sweep_benchmark_pgfplots.csv};

\addplot+[forget plot, mark=square*, color=cSSNStageOne]
    table[x=Eta, y=SSNStage1DOADeg, col sep=comma] {numerical_results/non_coherent_eta_sweep_benchmark_pgfplots.csv};

\addplot+[forget plot, mark=diamond*, color=cSSNStageTwo]
    table[x=Eta, y=SSNStage2DOADeg, col sep=comma] {numerical_results/non_coherent_eta_sweep_benchmark_pgfplots.csv};

\addplot+[forget plot, mark=triangle*, color=cTransMUSIC]
    table[x=Eta, y=TransMUSICStage2DOADeg, col sep=comma] {numerical_results/non_coherent_eta_sweep_benchmark_pgfplots.csv};

\addplot+[forget plot, mark=o, color=cDNNComplex]
    table[x=Eta, y=DNNComplexStage2DOADeg, col sep=comma] {numerical_results/non_coherent_eta_sweep_benchmark_pgfplots.csv};
\end{axis}
\end{tikzpicture}
\caption{DoA RMSPE.}
\end{subfigure}\hfill
\begin{subfigure}[t]{0.32\textwidth}
\centering
\begin{tikzpicture}
\begin{axis}[benchmarkSubfigureAxis,
xlabel={Array perturbation $\eta$},
    xmin=-0.0012,
    xmax=0.0312,
    xtick={0,0.005,0.01,0.015,0.02,0.025,0.03},
    ylabel={Uncertainty loss},
    ymode=log,]
\addplot+[forget plot, mark=none, dashed, very thick, color=cCCRB]
    table[x=Eta, y=CCRBRefUELoss, col sep=comma] {numerical_results/non_coherent_eta_sweep_benchmark_pgfplots.csv};

\addplot+[forget plot, mark=*, color=cESPRIT]
    table[x=Eta, y=ESPRITNetUELoss, col sep=comma] {numerical_results/non_coherent_eta_sweep_benchmark_pgfplots.csv};

\addplot+[forget plot, mark=square*, color=cSSNStageOne]
    table[x=Eta, y=SSNStage1NetUELoss, col sep=comma] {numerical_results/non_coherent_eta_sweep_benchmark_pgfplots.csv};

\addplot+[forget plot, mark=diamond*, color=cSSNStageTwo]
    table[x=Eta, y=SSNStage2NetUELoss, col sep=comma] {numerical_results/non_coherent_eta_sweep_benchmark_pgfplots.csv};

\addplot+[forget plot, mark=triangle*, color=cTransMUSIC]
    table[x=Eta, y=TransMUSICStage2NetUELoss, col sep=comma] {numerical_results/non_coherent_eta_sweep_benchmark_pgfplots.csv};

\addplot+[forget plot, mark=o, color=cDNNComplex]
    table[x=Eta, y=DNNComplexStage2NetUELoss, col sep=comma] {numerical_results/non_coherent_eta_sweep_benchmark_pgfplots.csv};
\end{axis}
\end{tikzpicture}
\caption{Uncertainty loss.}
\end{subfigure}\hfill
\begin{subfigure}[t]{0.32\textwidth}
\centering
\begin{tikzpicture}
\begin{axis}[benchmarkSubfigureAxis,
xlabel={Array perturbation $\eta$},
    xmin=-0.0012,
    xmax=0.0312,
    xtick={0,0.005,0.01,0.015,0.02,0.025,0.03},
    ylabel={Log norm. ANEES},
    ymajorgrids=true,]
\addplot+[forget plot, mark=none, dashed, very thick, color=cCCRB] coordinates {(-0.0012,0) (0.0312,0)};

\addplot+[forget plot, mark=*, color=cESPRIT]
    table[x=Eta, y=ESPRITLogANEESNormalized, col sep=comma] {numerical_results/non_coherent_eta_sweep_benchmark_pgfplots.csv};

\addplot+[forget plot, mark=square*, color=cSSNStageOne]
    table[x=Eta, y=SSNStage1LogANEESNormalized, col sep=comma] {numerical_results/non_coherent_eta_sweep_benchmark_pgfplots.csv};

\addplot+[forget plot, mark=diamond*, color=cSSNStageTwo]
    table[x=Eta, y=SSNStage2LogANEESNormalized, col sep=comma] {numerical_results/non_coherent_eta_sweep_benchmark_pgfplots.csv};

\addplot+[forget plot, mark=triangle*, color=cTransMUSIC]
    table[x=Eta, y=TransMUSICStage2LogANEESNormalized, col sep=comma] {numerical_results/non_coherent_eta_sweep_benchmark_pgfplots.csv};

\addplot+[forget plot, mark=o, color=cDNNComplex]
    table[x=Eta, y=DNNComplexStage2LogANEESNormalized, col sep=comma] {numerical_results/non_coherent_eta_sweep_benchmark_pgfplots.csv};
\end{axis}
\end{tikzpicture}
\caption{Log normalized ANEES.}
\end{subfigure}

\caption{Non-coherent sources, performance measures versus  array-perturbation ($\eta$)}
\label{fig:non_coherent_eta_benchmark_combined}
\label{fig:non_coherent_eta_benchmark_doa}
\label{fig:non_coherent_eta_benchmark_uncertainty}
\label{fig:non_coherent_eta_benchmark_log_anees}

\end{figure*}

%% file: plotting_latex_code/non_coherent_M_sweep_figures_pgfplots.tex

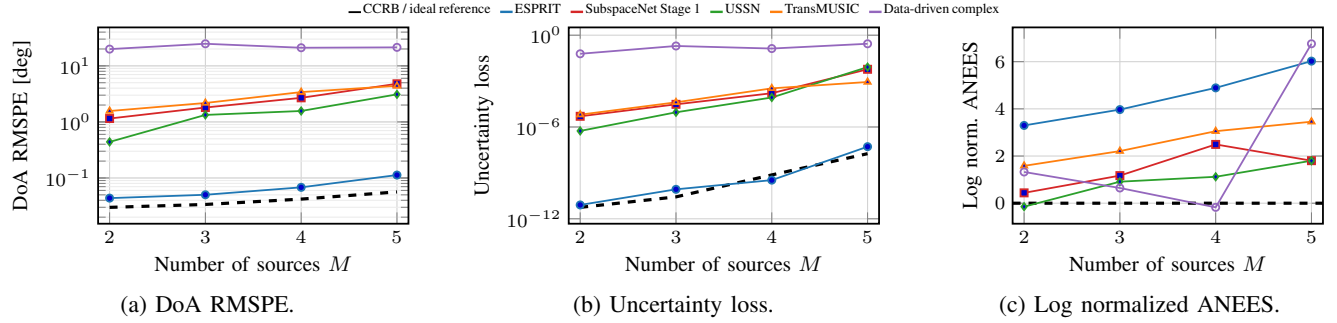
\begin{figure*}[!t]
\centering
\scriptsize
{\fontsize{5.2}{6.0}\selectfont
\makebox[\textwidth][c]{%
\textcolor{cCCRB}{\rule{0.9em}{0.7pt}}\,CCRB / ideal reference\hspace{0.70em}%
\textcolor{cESPRIT}{\rule{0.9em}{0.9pt}}\,ESPRIT\hspace{0.70em}%
\textcolor{cSSNStageOne}{\rule{0.9em}{0.9pt}}\,SubspaceNet Stage~1\hspace{0.70em}%
\textcolor{cSSNStageTwo}{\rule{0.9em}{0.9pt}}\,USSN\hspace{0.70em}%
\textcolor{cTransMUSIC}{\rule{0.9em}{0.9pt}}\,TransMUSIC\hspace{0.70em}%
\textcolor{cDNNComplex}{\rule{0.9em}{0.9pt}}\,Data-driven complex%
}
\par}

\vspace{0.25em}
\begin{subfigure}[t]{0.32\textwidth}
\centering
\begin{tikzpicture}
\begin{axis}[benchmarkSubfigureAxis,
xlabel={Number of sources $M$},
    xmin=1.88,
    xmax=5.12,
    xtick={2,3,4,5},
    ylabel={DoA RMSPE [deg]},
    ymode=log,]
\addplot+[forget plot, mark=none, dashed, very thick, color=cCCRB]
    table[x=SourcesM, y=CCRBSigmaDeg, col sep=comma] {numerical_results/non_coherent_M_sweep_benchmark_pgfplots.csv};

\addplot+[forget plot, mark=*, color=cESPRIT]
    table[x=SourcesM, y=ESPRITDOADeg, col sep=comma] {numerical_results/non_coherent_M_sweep_benchmark_pgfplots.csv};

\addplot+[forget plot, mark=square*, color=cSSNStageOne]
    table[x=SourcesM, y=SSNStage1DOADeg, col sep=comma] {numerical_results/non_coherent_M_sweep_benchmark_pgfplots.csv};

\addplot+[forget plot, mark=diamond*, color=cSSNStageTwo]
    table[x=SourcesM, y=SSNStage2DOADeg, col sep=comma] {numerical_results/non_coherent_M_sweep_benchmark_pgfplots.csv};

\addplot+[forget plot, mark=triangle*, color=cTransMUSIC]
    table[x=SourcesM, y=TransMUSICStage2DOADeg, col sep=comma] {numerical_results/non_coherent_M_sweep_benchmark_pgfplots.csv};

\addplot+[forget plot, mark=o, color=cDNNComplex]
    table[x=SourcesM, y=DNNComplexStage2DOADeg, col sep=comma] {numerical_results/non_coherent_M_sweep_benchmark_pgfplots.csv};
\end{axis}
\end{tikzpicture}
\caption{DoA RMSPE.}
\end{subfigure}\hfill
\begin{subfigure}[t]{0.32\textwidth}
\centering
\begin{tikzpicture}
\begin{axis}[benchmarkSubfigureAxis,
xlabel={Number of sources $M$},
    xmin=1.88,
    xmax=5.12,
    xtick={2,3,4,5},
    ylabel={Uncertainty loss},
    ymode=log,]
\addplot+[forget plot, mark=none, dashed, very thick, color=cCCRB]
    table[x=SourcesM, y=CCRBRefUELoss, col sep=comma] {numerical_results/non_coherent_M_sweep_benchmark_pgfplots.csv};

\addplot+[forget plot, mark=*, color=cESPRIT]
    table[x=SourcesM, y=ESPRITNetUELoss, col sep=comma] {numerical_results/non_coherent_M_sweep_benchmark_pgfplots.csv};

\addplot+[forget plot, mark=square*, color=cSSNStageOne]
    table[x=SourcesM, y=SSNStage1NetUELoss, col sep=comma] {numerical_results/non_coherent_M_sweep_benchmark_pgfplots.csv};

\addplot+[forget plot, mark=diamond*, color=cSSNStageTwo]
    table[x=SourcesM, y=SSNStage2NetUELoss, col sep=comma] {numerical_results/non_coherent_M_sweep_benchmark_pgfplots.csv};

\addplot+[forget plot, mark=triangle*, color=cTransMUSIC]
    table[x=SourcesM, y=TransMUSICStage2NetUELoss, col sep=comma] {numerical_results/non_coherent_M_sweep_benchmark_pgfplots.csv};

\addplot+[forget plot, mark=o, color=cDNNComplex]
    table[x=SourcesM, y=DNNComplexStage2NetUELoss, col sep=comma] {numerical_results/non_coherent_M_sweep_benchmark_pgfplots.csv};
\end{axis}
\end{tikzpicture}
\caption{Uncertainty loss.}
\end{subfigure}\hfill
\begin{subfigure}[t]{0.32\textwidth}
\centering
\begin{tikzpicture}
\begin{axis}[benchmarkSubfigureAxis,
xlabel={Number of sources $M$},
    xmin=1.88,
    xmax=5.12,
    xtick={2,3,4,5},
    ylabel={Log norm. ANEES},
    ymajorgrids=true,]
\addplot+[forget plot, mark=none, dashed, very thick, color=cCCRB] coordinates {(1.88,0) (5.12,0)};

\addplot+[forget plot, mark=*, color=cESPRIT]
    table[x=SourcesM, y=ESPRITLogANEESNormalized, col sep=comma] {numerical_results/non_coherent_M_sweep_benchmark_pgfplots.csv};

\addplot+[forget plot, mark=square*, color=cSSNStageOne]
    table[x=SourcesM, y=SSNStage1LogANEESNormalized, col sep=comma] {numerical_results/non_coherent_M_sweep_benchmark_pgfplots.csv};

\addplot+[forget plot, mark=diamond*, color=cSSNStageTwo]
    table[x=SourcesM, y=SSNStage2LogANEESNormalized, col sep=comma] {numerical_results/non_coherent_M_sweep_benchmark_pgfplots.csv};

\addplot+[forget plot, mark=triangle*, color=cTransMUSIC]
    table[x=SourcesM, y=TransMUSICStage2LogANEESNormalized, col sep=comma] {numerical_results/non_coherent_M_sweep_benchmark_pgfplots.csv};

\addplot+[forget plot, mark=o, color=cDNNComplex]
    table[x=SourcesM, y=DNNComplexStage2LogANEESNormalized, col sep=comma] {numerical_results/non_coherent_M_sweep_benchmark_pgfplots.csv};
\end{axis}
\end{tikzpicture}
\caption{Log normalized ANEES.}
\end{subfigure}

\caption{Non-coherent sources, performance measures versus source count ($M$)}
\label{fig:non_coherent_m_benchmark_combined}
\label{fig:non_coherent_m_benchmark_doa}
\label{fig:non_coherent_m_benchmark_uncertainty}
\label{fig:non_coherent_m_benchmark_log_anees}

\end{figure*}

%% file: plotting_latex_code/fixed_anchor_sweep_snr0.tex
\input{plotting_latex_code/fixed_anchor_sweep_snr0_accuracy.tex}
\input{plotting_latex_code/fixed_anchor_sweep_snr0_variance.tex}

%% file: plotting_latex_code/fixed_anchor_sweep_snr0_accuracy.tex
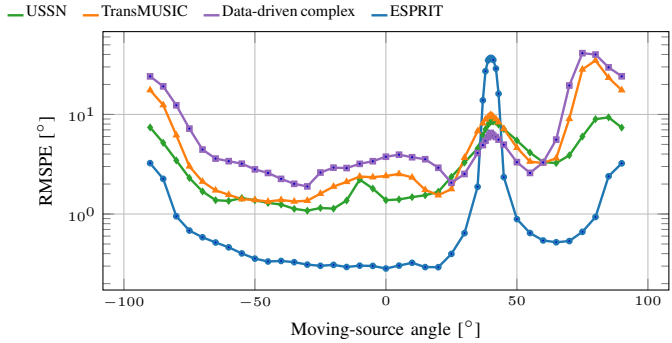
\begin{figure}[!t]
\scriptsize
{\fontsize{6}{7}\selectfont\makebox[\textwidth][l]{\textcolor{cSSNStageTwo}{\rule{1em}{0.9pt}}\,USSN\hspace{1em}\textcolor{cTransMUSIC}{\rule{1em}{0.9pt}}\,TransMUSIC\hspace{1em}\textcolor{cDNNComplex}{\rule{1em}{0.9pt}}\,Data-driven complex\hspace{1em}\textcolor{cESPRIT}{\rule{1em}{0.9pt}}\,ESPRIT}}
\par\vspace{0.6em}
\begin{minipage}{0.5\textwidth}
\centering
\begin{tikzpicture}
\begin{axis}[angleSweepZeroAxis,xlabel={Moving-source angle [$^\circ$]},ylabel={RMSPE [$^\circ$]},ymode=log]
\addplot+[forget plot,color=cSSNStageTwo,mark=diamond*,mark size=0.8pt,line width=0.9pt] table[col sep=comma,x=SweepAngleDeg,y=USSNAccuracyDeg] {numerical_results/fixed_anchor_sweep_snr0.csv};
\addplot+[forget plot,color=cTransMUSIC,mark=triangle*,mark size=0.8pt,line width=0.9pt] table[col sep=comma,x=SweepAngleDeg,y=TransMUSICAccuracyDeg] {numerical_results/fixed_anchor_sweep_snr0.csv};
\addplot+[forget plot,color=cDNNComplex,mark=square*,mark size=0.8pt,line width=0.9pt] table[col sep=comma,x=SweepAngleDeg,y=DataDrivenAccuracyDeg] {numerical_results/fixed_anchor_sweep_snr0.csv};
\addplot+[forget plot,color=cESPRIT,mark=*,mark size=0.8pt,line width=0.9pt] table[col sep=comma,x=SweepAngleDeg,y=ESPRITAccuracyDeg] {numerical_results/fixed_anchor_sweep_snr0.csv};
\end{axis}
\end{tikzpicture}
\end{minipage}
\caption{DoA Accuracy of the moving source}
\label{fig:all_models_fixed_anchor_sweep_snr0_accuracy}
\vspace{-0.2cm}
\end{figure}

%% file: plotting_latex_code/fixed_anchor_sweep_snr0_variance.tex
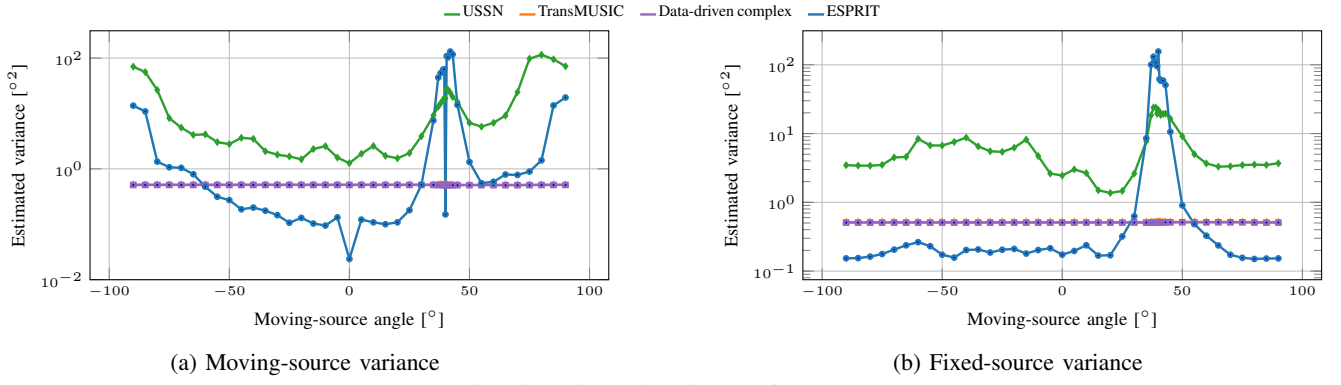
\begin{figure*}[!t]
\centering\scriptsize
{\fontsize{6}{7}\selectfont\makebox[\textwidth][c]{\textcolor{cSSNStageTwo}{\rule{1em}{0.9pt}}\,USSN\hspace{1em}\textcolor{cTransMUSIC}{\rule{1em}{0.9pt}}\,TransMUSIC\hspace{1em}\textcolor{cDNNComplex}{\rule{1em}{0.9pt}}\,Data-driven complex\hspace{1em}\textcolor{cESPRIT}{\rule{1em}{0.9pt}}\,ESPRIT}}
\par\vspace{0.6em}
\begin{subfigure}[t]{0.48\textwidth}
\centering
\begin{tikzpicture}
\begin{axis}[angleSweepZeroAxis,width=0.97\linewidth, height=0.56\linewidth,  xlabel={Moving-source angle [$^\circ$]},ylabel={Estimated variance [$^\circ{}^2$]},ymode=log]
\addplot+[forget plot,color=cSSNStageTwo,mark=diamond*,mark size=0.8pt,line width=0.9pt] table[col sep=comma,x=SweepAngleDeg,y=USSNMovingVarianceDeg2] {numerical_results/fixed_anchor_sweep_snr0.csv};
\addplot+[forget plot,color=cTransMUSIC,mark=triangle*,mark size=0.8pt,line width=0.9pt] table[col sep=comma,x=SweepAngleDeg,y=TransMUSICMovingVarianceDeg2] {numerical_results/fixed_anchor_sweep_snr0.csv};
\addplot+[forget plot,color=cDNNComplex,mark=square*,mark size=0.8pt,line width=0.9pt] table[col sep=comma,x=SweepAngleDeg,y=DataDrivenMovingVarianceDeg2] {numerical_results/fixed_anchor_sweep_snr0.csv};
\addplot+[forget plot,color=cESPRIT,mark=*,mark size=0.8pt,line width=0.9pt] table[col sep=comma,x=SweepAngleDeg,y=ESPRITMovingVarianceDeg2] {numerical_results/fixed_anchor_sweep_snr0.csv};
\end{axis}
\end{tikzpicture}
\caption{Moving-source variance}
\end{subfigure}\hfill
\begin{subfigure}[t]{0.48\textwidth}
\centering
\begin{tikzpicture}
\begin{axis}[angleSweepZeroAxis,width=0.97\linewidth, height=0.56\linewidth, xlabel={Moving-source angle [$^\circ$]},ylabel={Estimated variance [$^\circ{}^2$]},ymode=log]
\addplot+[forget plot,color=cSSNStageTwo,mark=diamond*,mark size=0.8pt,line width=0.9pt] table[col sep=comma,x=SweepAngleDeg,y=USSNAnchorVarianceDeg2] {numerical_results/fixed_anchor_sweep_snr0.csv};
\addplot+[forget plot,color=cTransMUSIC,mark=triangle*,mark size=0.8pt,line width=0.9pt] table[col sep=comma,x=SweepAngleDeg,y=TransMUSICAnchorVarianceDeg2] {numerical_results/fixed_anchor_sweep_snr0.csv};
\addplot+[forget plot,color=cDNNComplex,mark=square*,mark size=0.8pt,line width=0.9pt] table[col sep=comma,x=SweepAngleDeg,y=DataDrivenAnchorVarianceDeg2] {numerical_results/fixed_anchor_sweep_snr0.csv};
\addplot+[forget plot,color=cESPRIT,mark=*,mark size=0.8pt,line width=0.9pt] table[col sep=comma,x=SweepAngleDeg,y=ESPRITAnchorVarianceDeg2] {numerical_results/fixed_anchor_sweep_snr0.csv};
\end{axis}
\end{tikzpicture}
\caption{Fixed-source variance}
\end{subfigure}
\caption{Estimated variance of sources}
\label{fig:all_models_fixed_anchor_sweep_snr0_variance}
\end{figure*}

%% file: plotting_latex_code/modulation_16qam_snr_sweep_figures_pgfplots.tex

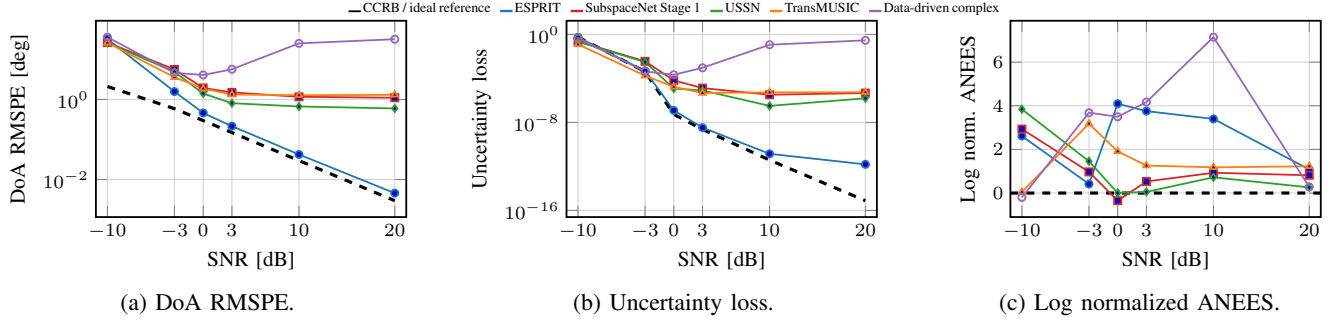
\begin{figure*}[!t]
\centering
\scriptsize
{\fontsize{5.2}{6.0}\selectfont
\makebox[\textwidth][c]{%
\textcolor{cCCRB}{\rule{0.9em}{0.7pt}}\,CCRB / ideal reference\hspace{0.70em}%
\textcolor{cESPRIT}{\rule{0.9em}{0.9pt}}\,ESPRIT\hspace{0.70em}%
\textcolor{cSSNStageOne}{\rule{0.9em}{0.9pt}}\,SubspaceNet Stage~1\hspace{0.70em}%
\textcolor{cSSNStageTwo}{\rule{0.9em}{0.9pt}}\,USSN\hspace{0.70em}%
\textcolor{cTransMUSIC}{\rule{0.9em}{0.9pt}}\,TransMUSIC\hspace{0.70em}%
\textcolor{cDNNComplex}{\rule{0.9em}{0.9pt}}\,Data-driven complex%
}
\par}

\vspace{0.25em}
\begin{subfigure}[t]{0.32\textwidth}
\centering
\begin{tikzpicture}
\begin{axis}[benchmarkSubfigureAxis,
xlabel={SNR [dB]},
    xmin=-11.2,
    xmax=21.2,
    xtick={-10,-3,0,3,10,20},
    ylabel={DoA RMSPE [deg]},
    ymode=log,]
\addplot+[forget plot, mark=none, dashed, very thick, color=cCCRB]
    table[x=SNRdB, y=CCRBSigmaDeg, col sep=comma] {numerical_results/modulation_16qam_snr_sweep_benchmark_pgfplots.csv};

\addplot+[forget plot, mark=*, color=cESPRIT]
    table[x=SNRdB, y=ESPRITDOADeg, col sep=comma] {numerical_results/modulation_16qam_snr_sweep_benchmark_pgfplots.csv};

\addplot+[forget plot, mark=square*, color=cSSNStageOne]
    table[x=SNRdB, y=SSNStage1DOADeg, col sep=comma] {numerical_results/modulation_16qam_snr_sweep_benchmark_pgfplots.csv};

\addplot+[forget plot, mark=diamond*, color=cSSNStageTwo]
    table[x=SNRdB, y=SSNStage2DOADeg, col sep=comma] {numerical_results/modulation_16qam_snr_sweep_benchmark_pgfplots.csv};

\addplot+[forget plot, mark=triangle*, color=cTransMUSIC]
    table[x=SNRdB, y=TransMUSICStage2DOADeg, col sep=comma] {numerical_results/modulation_16qam_snr_sweep_benchmark_pgfplots.csv};

\addplot+[forget plot, mark=o, color=cDNNComplex]
    table[x=SNRdB, y=DNNComplexStage2DOADeg, col sep=comma] {numerical_results/modulation_16qam_snr_sweep_benchmark_pgfplots.csv};
\end{axis}
\end{tikzpicture}
\caption{DoA RMSPE.}
\end{subfigure}\hfill
\begin{subfigure}[t]{0.32\textwidth}
\centering
\begin{tikzpicture}
\begin{axis}[benchmarkSubfigureAxis,
xlabel={SNR [dB]},
    xmin=-11.2,
    xmax=21.2,
    xtick={-10,-3,0,3,10,20},
    ylabel={Uncertainty loss},
    ymode=log,]
\addplot+[forget plot, mark=none, dashed, very thick, color=cCCRB]
    table[x=SNRdB, y=CCRBRefUELoss, col sep=comma] {numerical_results/modulation_16qam_snr_sweep_benchmark_pgfplots.csv};

\addplot+[forget plot, mark=*, color=cESPRIT]
    table[x=SNRdB, y=ESPRITNetUELoss, col sep=comma] {numerical_results/modulation_16qam_snr_sweep_benchmark_pgfplots.csv};

\addplot+[forget plot, mark=square*, color=cSSNStageOne]
    table[x=SNRdB, y=SSNStage1NetUELoss, col sep=comma] {numerical_results/modulation_16qam_snr_sweep_benchmark_pgfplots.csv};

\addplot+[forget plot, mark=diamond*, color=cSSNStageTwo]
    table[x=SNRdB, y=SSNStage2NetUELoss, col sep=comma] {numerical_results/modulation_16qam_snr_sweep_benchmark_pgfplots.csv};

\addplot+[forget plot, mark=triangle*, color=cTransMUSIC]
    table[x=SNRdB, y=TransMUSICStage2NetUELoss, col sep=comma] {numerical_results/modulation_16qam_snr_sweep_benchmark_pgfplots.csv};

\addplot+[forget plot, mark=o, color=cDNNComplex]
    table[x=SNRdB, y=DNNComplexStage2NetUELoss, col sep=comma] {numerical_results/modulation_16qam_snr_sweep_benchmark_pgfplots.csv};
\end{axis}
\end{tikzpicture}
\caption{Uncertainty loss.}
\end{subfigure}\hfill
\begin{subfigure}[t]{0.32\textwidth}
\centering
\begin{tikzpicture}
\begin{axis}[benchmarkSubfigureAxis,
xlabel={SNR [dB]},
    xmin=-11.2,
    xmax=21.2,
    xtick={-10,-3,0,3,10,20},
    ylabel={Log norm. ANEES},
    ymajorgrids=true,]
\addplot+[forget plot, mark=none, dashed, very thick, color=cCCRB] coordinates {(-11.2,0) (21.2,0)};

\addplot+[forget plot, mark=*, color=cESPRIT]
    table[x=SNRdB, y=ESPRITLogANEESNormalized, col sep=comma] {numerical_results/modulation_16qam_snr_sweep_benchmark_pgfplots.csv};

\addplot+[forget plot, mark=square*, color=cSSNStageOne]
    table[x=SNRdB, y=SSNStage1LogANEESNormalized, col sep=comma] {numerical_results/modulation_16qam_snr_sweep_benchmark_pgfplots.csv};

\addplot+[forget plot, mark=diamond*, color=cSSNStageTwo]
    table[x=SNRdB, y=SSNStage2LogANEESNormalized, col sep=comma] {numerical_results/modulation_16qam_snr_sweep_benchmark_pgfplots.csv};

\addplot+[forget plot, mark=triangle*, color=cTransMUSIC]
    table[x=SNRdB, y=TransMUSICStage2LogANEESNormalized, col sep=comma] {numerical_results/modulation_16qam_snr_sweep_benchmark_pgfplots.csv};

\addplot+[forget plot, mark=o, color=cDNNComplex]
    table[x=SNRdB, y=DNNComplexStage2LogANEESNormalized, col sep=comma] {numerical_results/modulation_16qam_snr_sweep_benchmark_pgfplots.csv};
\end{axis}
\end{tikzpicture}
\caption{Log normalized ANEES.}
\end{subfigure}

\caption{16-QAM-modulation sources, performance measures versus  SNR}
\label{fig:modulation_16qam_snr_benchmark_combined}
\label{fig:modulation_16qam_snr_benchmark_doa}
\label{fig:modulation_16qam_snr_benchmark_uncertainty}
\label{fig:modulation_16qam_snr_benchmark_log_anees}

\end{figure*}

%% file: plotting_latex_code/coherent_snr_sweep_figures_pgfplots.tex


\begin{figure*}[!t]
\centering
\scriptsize
{\fontsize{5.2}{6.0}\selectfont
\makebox[\textwidth][c]{%
\textcolor{cCCRB}{\rule{0.9em}{0.7pt}}\, Ideal reference\hspace{0.70em}%
\textcolor{cESPRIT}{\rule{0.9em}{0.9pt}}\,ESPRIT\hspace{0.70em}%
\textcolor{cSSNStageOne}{\rule{0.9em}{0.9pt}}\,SubspaceNet Stage~1\hspace{0.70em}%
\textcolor{cSSNStageTwo}{\rule{0.9em}{0.9pt}}\,USSN\hspace{0.70em}%
\textcolor{cTransMUSIC}{\rule{0.9em}{0.9pt}}\,TransMUSIC\hspace{0.70em}%
\textcolor{cDNNComplex}{\rule{0.9em}{0.9pt}}\,Data-driven complex%
}
\par}

\vspace{0.25em}
\begin{subfigure}[t]{0.32\textwidth}
\centering
\begin{tikzpicture}
\begin{axis}[benchmarkSubfigureAxis,
xlabel={SNR [dB]},
    xmin=-11.2,
    xmax=21.2,
    xtick={-10,-3,0,3,10,20},
    ylabel={DoA RMSPE [deg]},
    ymode=log,]
\addplot+[forget plot, mark=*, color=cESPRIT]
    table[x=SNRdB, y=ESPRITDOADeg, col sep=comma] {numerical_results/coherent_snr_sweep_benchmark_pgfplots.csv};

\addplot+[forget plot, mark=square*, color=cSSNStageOne]
    table[x=SNRdB, y=SSNStage1DOADeg, col sep=comma] {numerical_results/coherent_snr_sweep_benchmark_pgfplots.csv};

\addplot+[forget plot, mark=diamond*, color=cSSNStageTwo]
    table[x=SNRdB, y=SSNStage2DOADeg, col sep=comma] {numerical_results/coherent_snr_sweep_benchmark_pgfplots.csv};

\addplot+[forget plot, mark=triangle*, color=cTransMUSIC]
    table[x=SNRdB, y=TransMUSICStage2DOADeg, col sep=comma] {numerical_results/coherent_snr_sweep_benchmark_pgfplots.csv};

\addplot+[forget plot, mark=o, color=cDNNComplex]
    table[x=SNRdB, y=DNNComplexStage2DOADeg, col sep=comma] {numerical_results/coherent_snr_sweep_benchmark_pgfplots.csv};
\end{axis}
\end{tikzpicture}
\caption{DoA RMSPE.}
\end{subfigure}\hfill
\begin{subfigure}[t]{0.32\textwidth}
\centering
\begin{tikzpicture}
\begin{axis}[benchmarkSubfigureAxis,
xlabel={SNR [dB]},
    xmin=-11.2,
    xmax=21.2,
    xtick={-10,-3,0,3,10,20},
    ylabel={Uncertainty loss},
    ymode=log,]
\addplot+[forget plot, mark=*, color=cESPRIT]
    table[x=SNRdB, y=ESPRITNetUELoss, col sep=comma] {numerical_results/coherent_snr_sweep_benchmark_pgfplots.csv};

\addplot+[forget plot, mark=square*, color=cSSNStageOne]
    table[x=SNRdB, y=SSNStage1NetUELoss, col sep=comma] {numerical_results/coherent_snr_sweep_benchmark_pgfplots.csv};

\addplot+[forget plot, mark=diamond*, color=cSSNStageTwo]
    table[x=SNRdB, y=SSNStage2NetUELoss, col sep=comma] {numerical_results/coherent_snr_sweep_benchmark_pgfplots.csv};

\addplot+[forget plot, mark=triangle*, color=cTransMUSIC]
    table[x=SNRdB, y=TransMUSICStage2NetUELoss, col sep=comma] {numerical_results/coherent_snr_sweep_benchmark_pgfplots.csv};

\addplot+[forget plot, mark=o, color=cDNNComplex]
    table[x=SNRdB, y=DNNComplexStage2NetUELoss, col sep=comma] {numerical_results/coherent_snr_sweep_benchmark_pgfplots.csv};
\end{axis}
\end{tikzpicture}
\caption{Uncertainty loss.}
\end{subfigure}\hfill
\begin{subfigure}[t]{0.32\textwidth}
\centering
\begin{tikzpicture}
\begin{axis}[benchmarkSubfigureAxis,
xlabel={SNR [dB]},
    xmin=-11.2,
    xmax=21.2,
    xtick={-10,-3,0,3,10,20},
    ylabel={Log norm. ANEES},
    ymajorgrids=true,]
\addplot+[forget plot, mark=none, dashed, very thick, color=cCCRB] coordinates {(-11.2,0) (21.2,0)};

\addplot+[forget plot, mark=*, color=cESPRIT]
    table[x=SNRdB, y=ESPRITLogANEESNormalized, col sep=comma] {numerical_results/coherent_snr_sweep_benchmark_pgfplots.csv};

\addplot+[forget plot, mark=square*, color=cSSNStageOne]
    table[x=SNRdB, y=SSNStage1LogANEESNormalized, col sep=comma] {numerical_results/coherent_snr_sweep_benchmark_pgfplots.csv};

\addplot+[forget plot, mark=diamond*, color=cSSNStageTwo]
    table[x=SNRdB, y=SSNStage2LogANEESNormalized, col sep=comma] {numerical_results/coherent_snr_sweep_benchmark_pgfplots.csv};

\addplot+[forget plot, mark=triangle*, color=cTransMUSIC]
    table[x=SNRdB, y=TransMUSICStage2LogANEESNormalized, col sep=comma] {numerical_results/coherent_snr_sweep_benchmark_pgfplots.csv};

\addplot+[forget plot, mark=o, color=cDNNComplex]
    table[x=SNRdB, y=DNNComplexStage2LogANEESNormalized, col sep=comma] {numerical_results/coherent_snr_sweep_benchmark_pgfplots.csv};
\end{axis}
\end{tikzpicture}
\caption{Log normalized ANEES.}
\end{subfigure}

\caption{Coherent sources, performance measures versus  SNR}
\label{fig:coherent_snr_benchmark_combined}
\label{fig:coherent_snr_benchmark_doa}
\label{fig:coherent_snr_benchmark_uncertainty}
\label{fig:coherent_snr_benchmark_log_anees}

\end{figure*}
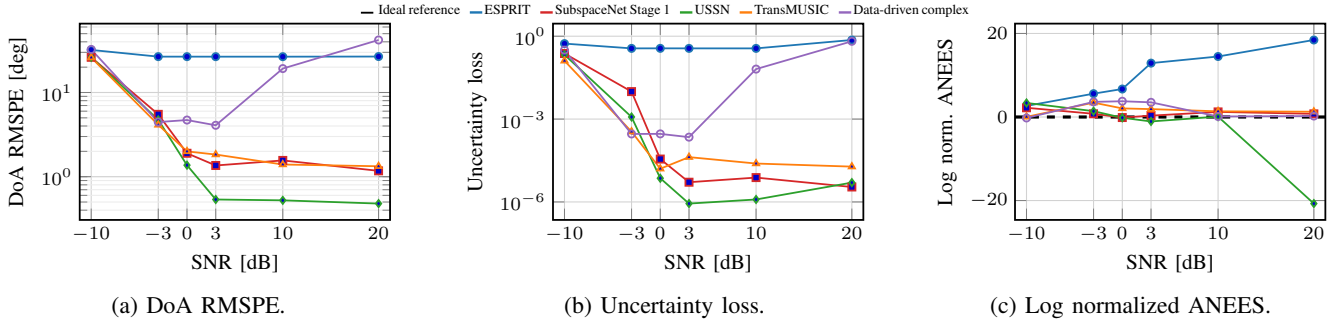

%% file: plotting_latex_code/coherent_eta_sweep_figures_pgfplots.tex


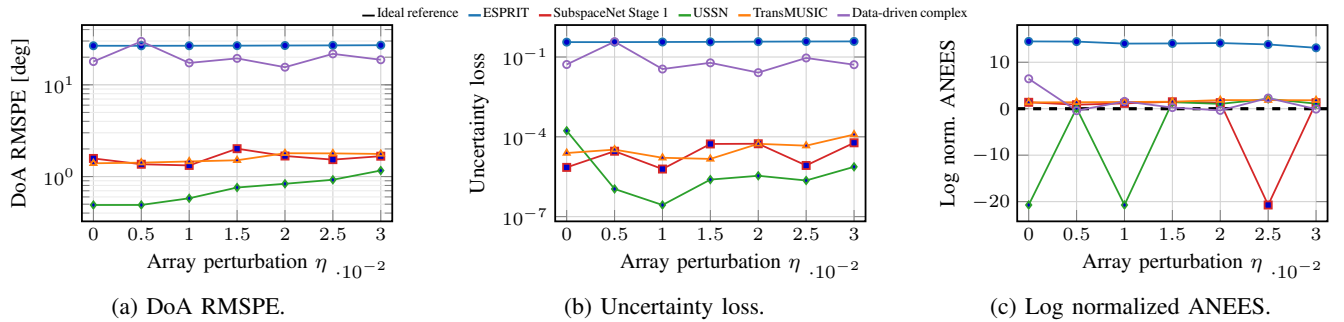
\begin{figure*}[!t]
\centering
\scriptsize
{\fontsize{5.2}{6.0}\selectfont
\makebox[\textwidth][c]{%
\textcolor{cCCRB}{\rule{0.9em}{0.7pt}}\,Ideal reference\hspace{0.70em}%
\textcolor{cESPRIT}{\rule{0.9em}{0.9pt}}\,ESPRIT\hspace{0.70em}%
\textcolor{cSSNStageOne}{\rule{0.9em}{0.9pt}}\,SubspaceNet Stage~1\hspace{0.70em}%
\textcolor{cSSNStageTwo}{\rule{0.9em}{0.9pt}}\,USSN\hspace{0.70em}%
\textcolor{cTransMUSIC}{\rule{0.9em}{0.9pt}}\,TransMUSIC\hspace{0.70em}%
\textcolor{cDNNComplex}{\rule{0.9em}{0.9pt}}\,Data-driven complex%
}
\par}

\vspace{0.25em}
\begin{subfigure}[t]{0.32\textwidth}
\centering
\begin{tikzpicture}
\begin{axis}[benchmarkSubfigureAxis,
xlabel={Array perturbation $\eta$},
    xmin=-0.0012,
    xmax=0.0312,
    xtick={0,0.005,0.01,0.015,0.02,0.025,0.03},
    ylabel={DoA RMSPE [deg]},
    ymode=log,]
\addplot+[forget plot, mark=*, color=cESPRIT]
    table[x=Eta, y=ESPRITDOADeg, col sep=comma] {numerical_results/coherent_eta_sweep_benchmark_pgfplots.csv};

\addplot+[forget plot, mark=square*, color=cSSNStageOne]
    table[x=Eta, y=SSNStage1DOADeg, col sep=comma] {numerical_results/coherent_eta_sweep_benchmark_pgfplots.csv};

\addplot+[forget plot, mark=diamond*, color=cSSNStageTwo]
    table[x=Eta, y=SSNStage2DOADeg, col sep=comma] {numerical_results/coherent_eta_sweep_benchmark_pgfplots.csv};

\addplot+[forget plot, mark=triangle*, color=cTransMUSIC]
    table[x=Eta, y=TransMUSICStage2DOADeg, col sep=comma] {numerical_results/coherent_eta_sweep_benchmark_pgfplots.csv};

\addplot+[forget plot, mark=o, color=cDNNComplex]
    table[x=Eta, y=DNNComplexStage2DOADeg, col sep=comma] {numerical_results/coherent_eta_sweep_benchmark_pgfplots.csv};
\end{axis}
\end{tikzpicture}
\caption{DoA RMSPE.}
\end{subfigure}\hfill
\begin{subfigure}[t]{0.32\textwidth}
\centering
\begin{tikzpicture}
\begin{axis}[benchmarkSubfigureAxis,
xlabel={Array perturbation $\eta$},
    xmin=-0.0012,
    xmax=0.0312,
    xtick={0,0.005,0.01,0.015,0.02,0.025,0.03},
    ylabel={Uncertainty loss},
    ymode=log,]
\addplot+[forget plot, mark=*, color=cESPRIT]
    table[x=Eta, y=ESPRITNetUELoss, col sep=comma] {numerical_results/coherent_eta_sweep_benchmark_pgfplots.csv};

\addplot+[forget plot, mark=square*, color=cSSNStageOne]
    table[x=Eta, y=SSNStage1NetUELoss, col sep=comma] {numerical_results/coherent_eta_sweep_benchmark_pgfplots.csv};

\addplot+[forget plot, mark=diamond*, color=cSSNStageTwo]
    table[x=Eta, y=SSNStage2NetUELoss, col sep=comma] {numerical_results/coherent_eta_sweep_benchmark_pgfplots.csv};

\addplot+[forget plot, mark=triangle*, color=cTransMUSIC]
    table[x=Eta, y=TransMUSICStage2NetUELoss, col sep=comma] {numerical_results/coherent_eta_sweep_benchmark_pgfplots.csv};

\addplot+[forget plot, mark=o, color=cDNNComplex]
    table[x=Eta, y=DNNComplexStage2NetUELoss, col sep=comma] {numerical_results/coherent_eta_sweep_benchmark_pgfplots.csv};
\end{axis}
\end{tikzpicture}
\caption{Uncertainty loss.}
\end{subfigure}\hfill
\begin{subfigure}[t]{0.32\textwidth}
\centering
\begin{tikzpicture}
\begin{axis}[benchmarkSubfigureAxis,
xlabel={Array perturbation $\eta$},
    xmin=-0.0012,
    xmax=0.0312,
    xtick={0,0.005,0.01,0.015,0.02,0.025,0.03},
    ylabel={Log norm. ANEES},
    ymajorgrids=true,]
\addplot+[forget plot, mark=none, dashed, very thick, color=cCCRB] coordinates {(-0.0012,0) (0.0312,0)};

\addplot+[forget plot, mark=*, color=cESPRIT]
    table[x=Eta, y=ESPRITLogANEESNormalized, col sep=comma] {numerical_results/coherent_eta_sweep_benchmark_pgfplots.csv};

\addplot+[forget plot, mark=square*, color=cSSNStageOne]
    table[x=Eta, y=SSNStage1LogANEESNormalized, col sep=comma] {numerical_results/coherent_eta_sweep_benchmark_pgfplots.csv};

\addplot+[forget plot, mark=diamond*, color=cSSNStageTwo]
    table[x=Eta, y=SSNStage2LogANEESNormalized, col sep=comma] {numerical_results/coherent_eta_sweep_benchmark_pgfplots.csv};

\addplot+[forget plot, mark=triangle*, color=cTransMUSIC]
    table[x=Eta, y=TransMUSICStage2LogANEESNormalized, col sep=comma] {numerical_results/coherent_eta_sweep_benchmark_pgfplots.csv};

\addplot+[forget plot, mark=o, color=cDNNComplex]
    table[x=Eta, y=DNNComplexStage2LogANEESNormalized, col sep=comma] {numerical_results/coherent_eta_sweep_benchmark_pgfplots.csv};
\end{axis}
\end{tikzpicture}
\caption{Log normalized ANEES.}
\end{subfigure}

\caption{Coherent sources, performance measures versus  array-perturbation ($\eta$)}
\label{fig:coherent_eta_benchmark_combined}
\label{fig:coherent_eta_benchmark_doa}
\label{fig:coherent_eta_benchmark_uncertainty}
\label{fig:coherent_eta_benchmark_log_anees}

\end{figure*}

%% file: plotting_latex_code/ssn_fulltrain_vs_cp_stage1_vs_stage2_uq_figures_pgfplots.tex
%

\definecolor{fullSOneBlue}{RGB}{31,119,180}
\definecolor{cpSOneRed}{RGB}{214,39,40}
\definecolor{stageTwoGreen}{RGB}{44,160,44}

\pgfplotsset{
  fulltraincpaxis/.style={
    width=0.97\linewidth,
    height=0.70\linewidth,
    grid=both,
    grid style={gray!20},
    xlabel={SNR [dB]},
    xtick={-10,-3,0,3,10,20},
    xmin=-11.5, xmax=21.5,
    tick label style={font=\scriptsize},
    label style={font=\footnotesize},
    title style={font=\footnotesize\bfseries},
    every axis plot/.append style={line width=0.95pt},
  },
  fullsone/.style={fullSOneBlue, mark=o, mark size=1.6pt},
  cpsone/.style={cpSOneRed, mark=triangle*, mark size=1.8pt},
  stagetwo/.style={stageTwoGreen, mark=square*, mark size=1.6pt},
  cptarget/.style={black, dashed, line width=0.75pt, mark=none},
}

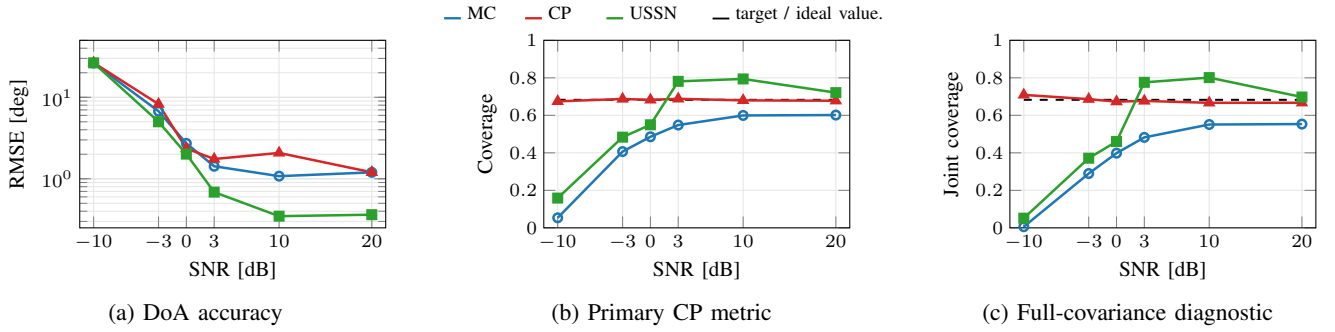
\begin{figure*}[tbp]
\centering
\scriptsize
\textcolor{fullSOneBlue}{\rule{0.9em}{0.9pt}} MC\hspace{1.2em}
\textcolor{cpSOneRed}{\rule{0.9em}{0.9pt}} CP \hspace{1.2em}
\textcolor{stageTwoGreen}{\rule{0.9em}{0.9pt}} USSN \hspace{1.2em}
\textcolor{black}{\rule{0.9em}{0.7pt}} target / ideal value.

\vspace{0.25em}
\begin{subfigure}[t]{0.32\textwidth}
\centering
\begin{tikzpicture}
\begin{axis}[fulltraincpaxis, ylabel={RMSE [deg]}, ymode=log, log basis y=10, ymin=0.25, ymax=50]
\addplot[fullsone] table[x=snr_db,y=full_stage1_mc_rmse_deg,col sep=comma]{numerical_results/ssn_fulltrain_vs_cp_stage1_vs_stage2_uq_plot_data.csv};
\addplot[cpsone] table[x=snr_db,y=cp_stage1_rmse_deg,col sep=comma]{numerical_results/ssn_fulltrain_vs_cp_stage1_vs_stage2_uq_plot_data.csv};
\addplot[stagetwo] table[x=snr_db,y=stage2_rmse_deg,col sep=comma]{numerical_results/ssn_fulltrain_vs_cp_stage1_vs_stage2_uq_plot_data.csv};
\end{axis}
\end{tikzpicture}
\caption{DoA accuracy}
\end{subfigure}\hfill
\begin{subfigure}[t]{0.32\textwidth}
\centering
\begin{tikzpicture}
\begin{axis}[fulltraincpaxis, ylabel={Coverage}, ymin=0, ymax=1]
\addplot[cptarget] table[x=snr_db,y=target_coverage,col sep=comma]{numerical_results/ssn_fulltrain_vs_cp_stage1_vs_stage2_uq_plot_data.csv};
\addplot[fullsone] table[x=snr_db,y=full_stage1_mc_source_coverage_target,col sep=comma]{numerical_results/ssn_fulltrain_vs_cp_stage1_vs_stage2_uq_plot_data.csv};
\addplot[cpsone] table[x=snr_db,y=cp_stage1_source_coverage_target,col sep=comma]{numerical_results/ssn_fulltrain_vs_cp_stage1_vs_stage2_uq_plot_data.csv};
\addplot[stagetwo] table[x=snr_db,y=stage2_source_coverage_target,col sep=comma]{numerical_results/ssn_fulltrain_vs_cp_stage1_vs_stage2_uq_plot_data.csv};
\end{axis}
\end{tikzpicture}
\caption{Primary CP metric}
\end{subfigure}\hfill
\begin{subfigure}[t]{0.32\textwidth}
\centering
\begin{tikzpicture}
\begin{axis}[fulltraincpaxis, ylabel={Joint coverage}, ymin=0, ymax=1]
\addplot[cptarget] table[x=snr_db,y=target_coverage,col sep=comma]{numerical_results/ssn_fulltrain_vs_cp_stage1_vs_stage2_uq_plot_data.csv};
\addplot[fullsone] table[x=snr_db,y=full_stage1_mc_coverage_chi2_target,col sep=comma]{numerical_results/ssn_fulltrain_vs_cp_stage1_vs_stage2_uq_plot_data.csv};
\addplot[cpsone] table[x=snr_db,y=cp_stage1_coverage_chi2_target,col sep=comma]{numerical_results/ssn_fulltrain_vs_cp_stage1_vs_stage2_uq_plot_data.csv};
\addplot[stagetwo] table[x=snr_db,y=stage2_coverage_chi2_target,col sep=comma]{numerical_results/ssn_fulltrain_vs_cp_stage1_vs_stage2_uq_plot_data.csv};
\end{axis}
\end{tikzpicture}
\caption{Full-covariance diagnostic}
\end{subfigure}
\caption{Accuracy and coverage for uncertainty extraction methods}
\label{fig:fulltrain-cp-accuracy-coverage}
\end{figure*}

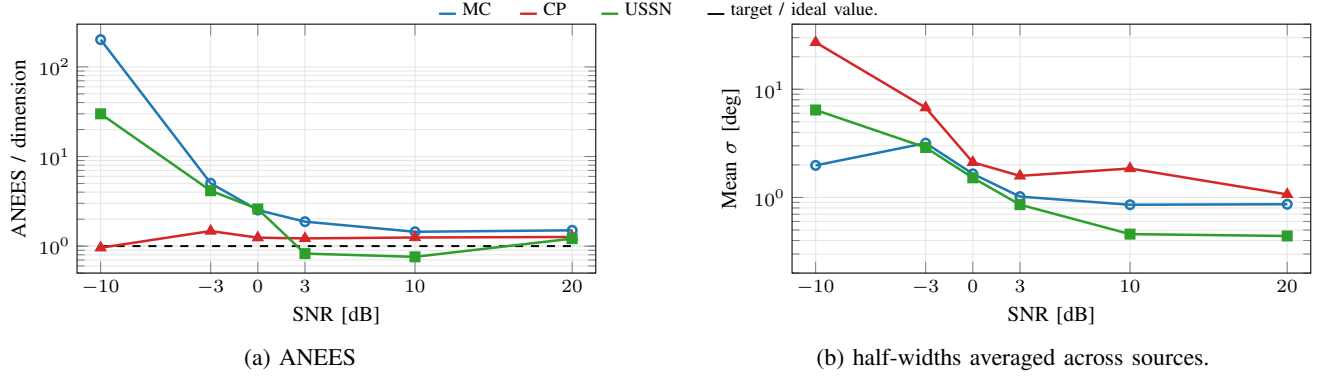
\begin{figure*}[tbp]
\centering
\scriptsize
\textcolor{fullSOneBlue}{\rule{0.9em}{0.9pt}} MC\hspace{1.2em}
\textcolor{cpSOneRed}{\rule{0.9em}{0.9pt}} CP \hspace{1.2em}
\textcolor{stageTwoGreen}{\rule{0.9em}{0.9pt}} USSN \hspace{1.2em}
\textcolor{black}{\rule{0.9em}{0.7pt}} target / ideal value.

\vspace{0.25em}
\begin{subfigure}[t]{0.48\textwidth}
\centering
\begin{tikzpicture}
\begin{axis}[fulltraincpaxis, width=0.97\linewidth, height=0.56\linewidth, ylabel={ANEES / dimension}, ymode=log, log basis y=10, ymin=0.5, ymax=300]
\addplot[cptarget] coordinates {(-10,1)(20,1)};
\addplot[fullsone] table[x=snr_db,y=full_stage1_mc_anees_normalized,col sep=comma]{numerical_results/ssn_fulltrain_vs_cp_stage1_vs_stage2_uq_plot_data.csv};
\addplot[cpsone] table[x=snr_db,y=cp_stage1_anees_normalized,col sep=comma]{numerical_results/ssn_fulltrain_vs_cp_stage1_vs_stage2_uq_plot_data.csv};
\addplot[stagetwo] table[x=snr_db,y=stage2_anees_normalized,col sep=comma]{numerical_results/ssn_fulltrain_vs_cp_stage1_vs_stage2_uq_plot_data.csv};
\end{axis}
\end{tikzpicture}
\caption{ANEES}
\end{subfigure}\hfill
\begin{subfigure}[t]{0.48\textwidth}
\centering
\begin{tikzpicture}
\begin{axis}[fulltraincpaxis, width=0.97\linewidth, height=0.56\linewidth, ylabel={Mean $\sigma$ [deg]}, ymode=log, log basis y=10, ymin=0.2, ymax=40]
\addplot[fullsone] table[x=snr_db,y=full_stage1_mc_mean_sigma_deg,col sep=comma]{numerical_results/ssn_fulltrain_vs_cp_stage1_vs_stage2_uq_plot_data.csv};
\addplot[cpsone] table[x=snr_db,y=cp_stage1_mean_sigma_deg,col sep=comma]{numerical_results/ssn_fulltrain_vs_cp_stage1_vs_stage2_uq_plot_data.csv};
\addplot[stagetwo] table[x=snr_db,y=stage2_mean_sigma_deg,col sep=comma]{numerical_results/ssn_fulltrain_vs_cp_stage1_vs_stage2_uq_plot_data.csv};
\end{axis}
\end{tikzpicture}
\caption{half-widths averaged across sources.}
\end{subfigure}\hfill
\caption{Uncertainty diagnostics of CP and USSN}
\label{fig:fulltrain-cp-uncertainty}
\end{figure*}

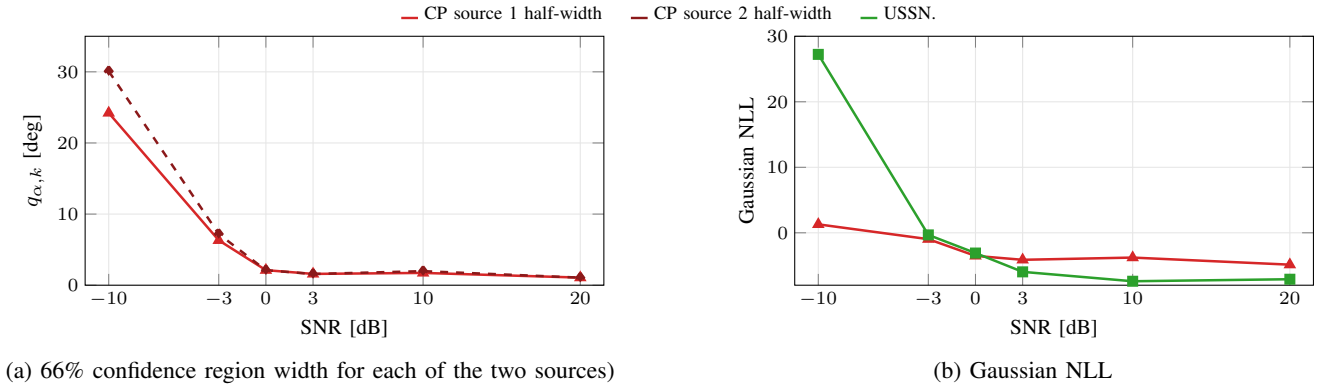
\begin{figure*}[tbp]
\centering
\scriptsize
\textcolor{cpSOneRed}{\rule{0.9em}{0.9pt}} CP source 1 half-width\hspace{1.2em}
\textcolor{cpSOneRed!65!black}{\rule{0.9em}{0.9pt}} CP source 2 half-width\hspace{1.2em}
\textcolor{stageTwoGreen}{\rule{0.9em}{0.9pt}} USSN.

\vspace{0.25em}
\begin{subfigure}[t]{0.48\textwidth}
\centering
\begin{tikzpicture}
\begin{axis}[fulltraincpaxis, width=0.97\linewidth, height=0.56\linewidth, ylabel={$q_{\alpha,k}$ [deg]}, ymin=0, ymax=35]
\addplot[cpsone] table[x=snr_db,y=cp_qhat_source1_deg,col sep=comma]{numerical_results/ssn_fulltrain_vs_cp_stage1_vs_stage2_uq_plot_data.csv};
\addplot[cpSOneRed!65!black, dashed, mark=diamond*, mark size=1.8pt] table[x=snr_db,y=cp_qhat_source2_deg,col sep=comma]{numerical_results/ssn_fulltrain_vs_cp_stage1_vs_stage2_uq_plot_data.csv};
\end{axis}
\end{tikzpicture}
\caption{66\% confidence region width for each of the two sources)}
\end{subfigure}\hfill
\begin{subfigure}[t]{0.48\textwidth}
\centering
\begin{tikzpicture}
\begin{axis}[fulltraincpaxis, width=0.97\linewidth, height=0.56\linewidth,  ylabel={Gaussian NLL}, ymin=-8, ymax=30, clip=false]
\addplot[cpsone] table[x=snr_db,y=cp_stage1_gaussian_nll,col sep=comma]{numerical_results/ssn_fulltrain_vs_cp_stage1_vs_stage2_uq_plot_data.csv};
\addplot[stagetwo] table[x=snr_db,y=stage2_gaussian_nll,col sep=comma]{numerical_results/ssn_fulltrain_vs_cp_stage1_vs_stage2_uq_plot_data.csv};
\end{axis}
\end{tikzpicture}
\caption{Gaussian NLL}
\end{subfigure}

\caption{The confidence intervals obtained with CP and the resulting Gaussian negative log-likelihood (NLL)}
\label{fig:fulltrain-cp-qhat-nll}
\end{figure*}

%% file: TSP_v01.bbl
\begin{thebibliography}{10}
\providecommand{\url}[1]{#1}
\csname url@samestyle\endcsname
\providecommand{\newblock}{\relax}
\providecommand{\bibinfo}[2]{#2}
\providecommand{\BIBentrySTDinterwordspacing}{\spaceskip=0pt\relax}
\providecommand{\BIBentryALTinterwordstretchfactor}{4}
\providecommand{\BIBentryALTinterwordspacing}{\spaceskip=\fontdimen2\font plus
\BIBentryALTinterwordstretchfactor\fontdimen3\font minus \fontdimen4\font\relax}
\providecommand{\BIBforeignlanguage}[2]{{%
\expandafter\ifx\csname l@#1\endcsname\relax
\typeout{** WARNING: IEEEtran.bst: No hyphenation pattern has been}%
\typeout{** loaded for the language `#1'. Using the pattern for}%
\typeout{** the default language instead.}%
\else
\language=\csname l@#1\endcsname
\fi
#2}}
\providecommand{\BIBdecl}{\relax}
\BIBdecl

\bibitem{zohar2026deep}
R.~Zohar, S.~Ginzach, and N.~Shlezinger, ``Deep learning-aided {DoA} estimation with uncertainty extraction,'' in \emph{IEEE Sensor Array and Multichannel Signal Processing Workshop (SAM)}, 2026.

\bibitem{pillai2012array}
S.~U. Pillai, \emph{Array signal processing}.\hskip 1em plus 0.5em minus 0.4em\relax Springer, 2012.

\bibitem{tuncer2009classical}
T.~E. Tuncer and B.~Friedlander, \emph{Classical and modern direction-of-arrival estimation}.\hskip 1em plus 0.5em minus 0.4em\relax Academic Press, 2009.

\bibitem{haykin1992some}
S.~Haykin, J.~Reilly, V.~Kezys, and E.~Vertatschitsch, ``Some aspects of array signal processing,'' in \emph{IEE Proceedings F (Radar and Signal Processing)}, vol. 139, no.~1.\hskip 1em plus 0.5em minus 0.4em\relax IET, 1992, pp. 1--26.

\bibitem{konstantino2026unsupervised}
S.~Konstantino, L.~Li, N.~Shlezinger, and D.~Dardari, ``Unsupervised adaptation of {AI DoA} estimators via downstream tracking,'' in \emph{Proc. IEEE ICASSP}, 2026, pp. 22\,352--22\,356.

\bibitem{hawkes2003wideband}
M.~Hawkes and A.~Nehorai, ``Wideband source localization using a distributed acoustic vector-sensor array,'' \emph{{IEEE} Trans. Signal Process.}, vol.~51, no.~6, pp. 1479--1491, 2003.

\bibitem{yang2020bayesian}
Y.~Yang, S.~Dang, M.~Wen, S.~Mumtaz, and M.~Guizani, ``Bayesian beamforming for mobile millimeter wave channel tracking in the presence of {DOA} uncertainty,'' \emph{{IEEE} Trans. Commun.}, vol.~68, no.~12, pp. 7547--7562, 2020.

\bibitem{lam2006bayesian}
C.~J. Lam and A.~C. Singer, ``Bayesian beamforming for {DOA} uncertainty: theory and implementation,'' \emph{{IEEE} Trans. Signal Process.}, vol.~54, no.~11, pp. 4435--4445, 2006.

\bibitem{capon1969mvdrbf}
J.~Capon, ``High-resolution frequency-wavenumber spectrum analysis,'' \emph{Proc. {IEEE}}, vol.~57, no.~8, pp. 1408--1418, 1969.

\bibitem{benesty2017fundamentals}
J.~Benesty, I.~Cohen, and J.~Chen, \emph{Fundamentals of signal enhancement and array signal processing}.\hskip 1em plus 0.5em minus 0.4em\relax Wiley, 2017.

\bibitem{schmidt1986music}
R.~Schmidt, ``Multiple emitter location and signal parameter estimation,'' \emph{{IEEE} Trans. Antennas Propag.}, vol.~34, no.~3, pp. 276--280, 1986.

\bibitem{Barabell1983ImprovingTR}
A.~J. Barabell, ``Improving the resolution performance of eigenstructure-based direction-finding algorithms,'' in \emph{Proc. IEEE ICASSP}, 1983.

\bibitem{roy1989esprit}
R.~Roy and T.~Kailath, ``{ESPRIT}-estimation of signal parameters via rotational invariance techniques,'' \emph{{IEEE} Trans. Acoust., Speech, Signal Process.}, vol.~37, no.~7, pp. 984--995, 1989.

\bibitem{liu2023twenty}
W.~Liu, M.~Haardt, M.~S. Greco, C.~F. Mecklenbr{\"a}uker, and P.~Willett, ``Twenty-five years of sensor array and multichannel signal processing: A review of progress to date and potential research directions,'' \emph{{IEEE} Signal Process. Mag.}, vol.~40, no.~4, pp. 80--91, 2023.

\bibitem{stoica1989music}
P.~Stoica and A.~Nehorai, ``{MUSIC}, maximum likelihood, and {C}ramer-{R}ao bound,'' \emph{{IEEE} Trans. Acoust., Speech, Signal Process.}, vol.~37, no.~5, pp. 720--741, 1989.

\bibitem{liang2020review}
Y.~Liang, W.~Liu, Q.~Shen, W.~Cui, and S.~Wu, ``A review of closed-form {C}ram{\'e}r-{R}ao bounds for {DOA} estimation in the presence of {G}aussian noise under a unified framework,'' \emph{{IEEE} Access}, vol.~8, pp. 175\,101--175\,124, 2020.

\bibitem{rao2002performance}
B.~D. Rao and K.~S. Hari, ``Performance analysis of root-{MUSIC},'' \emph{{IEEE} Trans. Acoust., Speech, Signal Process.}, vol.~37, no.~12, pp. 1939--1949, 1989.

\bibitem{yuen2002asymptotic}
N.~Yuen and B.~Friedlander, ``Asymptotic performance analysis of {ESPRIT}, higher order {ESPRIT}, and virtual {ESPRIT} algorithms,'' \emph{{IEEE} Trans. Signal Process.}, vol.~44, no.~10, pp. 2537--2550, 2002.

\bibitem{al2022review}
H.~Al~Kassir \emph{et~al.}, ``A review of the state of the art and future challenges of deep learning-based beamforming,'' \emph{{IEEE} Access}, vol.~10, pp. 80\,869--80\,882, 2022.

\bibitem{DNN_WITH_Antenna_ARRAY}
M.~Chen, Y.~Gong, and X.~Mao, ``Deep neural network for estimation of direction of arrival with antenna array,'' \emph{{IEEE} Access}, vol.~8, pp. 140\,688--140\,698, 2020.

\bibitem{cong2020robust}
J.~Cong, X.~Wang, M.~Huang, and L.~Wan, ``Robust {DOA} estimation method for {MIMO} radar via deep neural networks,'' \emph{{IEEE} Sensors J.}, vol.~21, no.~6, pp. 7498--7507, 2021.

\bibitem{feintuch2023neural}
S.~Feintuch \emph{et~al.}, ``Neural-network-based {DOA} estimation in the presence of non-gaussian interference,'' \emph{{IEEE} Trans. Aerosp. Electron. Syst.}, vol.~60, no.~1, pp. 119--132, 2023.

\bibitem{DOAEstimation_LowSNR}
G.~K. Papageorgiou, M.~Sellathurai, and Y.~C. Eldar, ``Deep networks for direction-of-arrival estimation in low {SNR},'' \emph{{IEEE} Trans. Signal Process.}, vol.~69, pp. 3714--3729, 2021.

\bibitem{lee2022deep}
H.~Lee \emph{et~al.}, ``Deep learning-based near-field source localization without a priori knowledge of the number of sources,'' \emph{{IEEE} Access}, vol.~10, pp. 55\,360--55\,368, 2022.

\bibitem{zheng2024deepdoa}
S.~Zheng \emph{et~al.}, ``Deep learning-based {DOA} estimation,'' \emph{{IEEE} Trans. on Cogn. Commun. Netw.}, vol.~10, no.~3, pp. 819--835, Jun. 2024.

\bibitem{lan2023novel}
X.~Lan \emph{et~al.}, ``A novel {DOA} estimation of closely spaced sources using attention mechanism with conformal arrays,'' \emph{{IEEE} Access}, vol.~11, pp. 44\,010--44\,018, 2023.

\bibitem{ji2024transmusic}
J.~Ji, W.~Mao, F.~Xi, and S.~Chen, ``Trans{MUSIC}: A transformer-aided subspace method for {DOA} estimation with low-resolution {ADC}s,'' in \emph{Proc, IEEE ICASSP}, 2024.

\bibitem{shlezinger2023model}
N.~Shlezinger and Y.~C. Eldar, ``Model-based deep learning,'' \emph{Foundations and Trends{\textregistered} in Signal Processing}, vol.~17, no.~4, pp. 291--416, 2023.

\bibitem{lee2022ftmr}
D.~T. Hoang and K.~Lee, ``Deep learning-aided coherent direction-of-arrival estimation with the {FTMR} algorithm,'' \emph{{IEEE} Trans. Signal Process.}, vol.~70, pp. 1118--1130, 2022.

\bibitem{elbir2020deepmusic}
A.~M. Elbir, ``{DeepMUSIC}: Multiple signal classification via deep learning,'' \emph{IEEE Sens. Lett.}, vol.~4, pp. 1--4, 2020.

\bibitem{barthelme2021doa}
A.~Barthelme and W.~Utschick, ``{DoA} estimation using neural network-based covariance matrix reconstruction,'' \emph{{IEEE} Signal Process. Lett.}, vol.~28, pp. 783--787, 2021.

\bibitem{wu2022gridless}
X.~Wu \emph{et~al.}, ``A gridless {DOA} estimation method based on convolutional neural network with {T}oeplitz prior,'' \emph{{IEEE} Signal Process. Lett.}, vol.~29, pp. 1247--1251, 2022.

\bibitem{jiang2023toeplitz}
Z.~Jiang \emph{et~al.}, ``A {T}oeplitz prior-based deep learning framework for {DOA} estimation with unknown mutual coupling,'' in \emph{Proc. EUSIPCO}, 2023.

\bibitem{shiran2026deep}
T.~Shiran, Y.~Gilady, O.~Poran, and N.~Shlezinger, ``Deep unfolded subspace-based {DoA} recovery from sparse arrays,'' in \emph{Proc. IEEE ICASSP}, 2026.

\bibitem{shmuel2023subspacenet}
D.~H. Shmuel \emph{et~al.}, ``Subspace{N}et: Deep learning-aided subspace methods for {DoA} estimation,'' \emph{{IEEE} Trans. Veh. Technol.}, vol.~74, no.~3, pp. 4962--4976, 2025.

\bibitem{DA-MUSIC-2023}
J.~P. Merkofer \emph{et~al.}, ``Data-driven {DoA} estimation via deep augmented {MUSIC} algorithm,'' \emph{{IEEE} Trans. Veh. Technol.}, vol.~73, no.~2, pp. 2771--2785, 2024.

\bibitem{xu2024md}
X.~Xu and Q.~Huang, ``{MD-DOA}: A model-based deep learning {DOA} estimation architecture,'' \emph{{IEEE} Sensors J.}, vol.~24, no.~12, pp. 20\,240--20\,253, 2024.

\bibitem{gast2025near}
A.~Gast, L.~L. Magoarou, and N.~Shlezinger, ``Near field localization via {AI}-aided subspace methods,'' \emph{arXiv preprint arXiv:2504.00599}, 2025.

\bibitem{zohar2025remote}
R.~Zohar, S.~Ginzach, and N.~Shlezinger, ``Remote {DoA} estimation via subspace-oriented deep-learning-aided vector quantization,'' in \emph{Proc. IEEE SPAWC}, 2025.

\bibitem{gawlikowski2023survey}
J.~Gawlikowski \emph{et~al.}, ``A survey of uncertainty in deep neural networks,'' \emph{Artificial Intelligence Review}, vol.~56, pp. 1513--1589, 2023.

\bibitem{jospin2022hands}
L.~V. Jospin, H.~Laga, F.~Boussaid, W.~Buntine, and M.~Bennamoun, ``Hands-on {B}ayesian neural networks—a tutorial for deep learning users,'' \emph{{IEEE} Comput. Intell. Mag.}, vol.~17, no.~2, pp. 29--48, 2022.

\bibitem{fu2026deep}
J.~Fu \emph{et~al.}, ``A deep learning-based {DOA} measurement method for underwater wideband sources using linear array with uncertainty quantification,'' \emph{{IEEE} Trans. Instrum. Meas.}, vol.~75, 2026.

\bibitem{rahaman2021uncertainty}
R.~Rahaman and A.~Thiery, ``Uncertainty quantification and deep ensembles,'' \emph{Advances in Neural Information Processing Systems}, vol.~34, pp. 20\,063--20\,075, 2021.

\bibitem{lindemann2024formal}
L.~Lindemann \emph{et~al.}, ``Formal verification and control with conformal prediction,'' \emph{{IEEE} Control Syst. Mag.}, vol.~45, no.~6, pp. 72--122, 2025.

\bibitem{khurjekar2023uncertainty}
I.~D. Khurjekar and P.~Gerstoft, ``Uncertainty quantification for direction-of-arrival estimation with conformal prediction,'' \emph{The Journal of the Acoustical Society of America}, vol. 154, no.~2, pp. 979--990, 2023.

\bibitem{rozenfeld2026uncertainty}
V.~Rozenfeld and B.~L. Goldshtein, ``Uncertainty quantification and risk control for multi-speaker sound source localization,'' \emph{arXiv preprint arXiv:2603.17377}, 2026.

\bibitem{shlezinger2020model}
N.~Shlezinger, J.~Whang, Y.~C. Eldar, and A.~G. Dimakis, ``Model-based deep learning,'' \emph{Proc. {IEEE}}, vol. 111, no.~5, pp. 465--499, 2023.

\bibitem{shlezinger2022discriminative}
N.~Shlezinger and T.~Routtenberg, ``Discriminative and generative learning for linear estimation of random signals,'' \emph{{IEEE} Signal Process. Mag.}, vol.~40, no.~6, pp. 75--82, 2023.

\bibitem{weisman2026conformal}
O.~Weisman, N.~Shlezinger, and B.~Laufer-Goldshtein, ``Conformal prediction aided {K}alman filters with confidence intervals,'' in \emph{Proc. IEEE ICASSP}, 2026.

\bibitem{dahan2025bayesian}
Y.~Dahan \emph{et~al.}, ``Bayesian {K}alman{N}et: Quantifying uncertainty in deep learning augmented {K}alman filter,'' \emph{{IEEE} Trans. Signal Process.}, vol.~73, pp. 2558--2573, 2025.

\bibitem{stewart1990matrix}
G.~W. Stewart and J.-g. Sun, \emph{Matrix perturbation theory}.\hskip 1em plus 0.5em minus 0.4em\relax Academic Press, 1990.

\bibitem{carrier2005functions}
G.~F. Carrier, M.~Krook, and C.~E. Pearson, \emph{Functions of a complex variable: theory and technique}.\hskip 1em plus 0.5em minus 0.4em\relax SIAM, 2005.

\end{thebibliography}
